\documentclass[11pt]{article}

\newif\ifnotes
\notestrue

\usepackage{hyperref}
\usepackage{fullpage}
\usepackage{amssymb}
\usepackage{amsmath}
\usepackage{amsthm}
\usepackage{amsfonts}
\usepackage{bm}
\usepackage{enumitem}
\usepackage{color}
\usepackage{comment}
\usepackage[capitalize]{cleveref}
\usepackage[dvipsnames]{xcolor}
\usepackage{float}
\usepackage[T1]{fontenc}
\usepackage{todonotes}
\usepackage{asymptote}
\usepackage{mdframed}
\usepackage[most]{tcolorbox}
\usepackage{hyperref}
\usepackage{enumitem}
\usepackage{framed}
\usepackage{mdframed}
\usepackage{scrextend}
\usepackage{bbm}
\usepackage{thm-restate}
\usepackage{mathtools}
\usetikzlibrary{arrows.meta}
\tikzset{>={Latex[width=2mm,length=2mm]}}

\newcommand{\mnote}[1]{\ifnotes $\ll$\textsf{\color{purple} Meghal: { #1}}$\gg$ \fi}
\notesfalse

\usepackage[T1]{fontenc}

\definecolor{halfgreen}{rgb}{0.0, 0.5, 0.25}

\usepackage{hyperref}
\hypersetup{
    colorlinks=true,
    linkcolor=blue,
    filecolor=blue,
    citecolor=red,
    urlcolor=red,
}
\usepackage[hyperpageref]{backref}

\newtheorem{theorem}{Theorem}[section]
\newtheorem{lemma}[theorem]{Lemma}

\newtheorem{conjecture}[theorem]{Conjecture}

\theoremstyle{definition}
\newtheorem{definition}[theorem]{Definition}

\Crefname{theorem}{Theorem}{Theorems}
\Crefname{claim}{Claim}{Claims}
\Crefname{lemma}{Lemma}{Lemmas}
\Crefname{proposition}{Proposition}{Propositions}
\Crefname{corollary}{Corollary}{Corollaries}
\Crefname{definition}{Definition}{Definitions}

\newcommand{\C}{\mathsf{C}}

\newcommand{\del}{\mathsf{del}}
\newcommand{\dec}{\mathsf{query}}
\newcommand{\err}{\mathsf{err}}
\newcommand{\poly}{\text{poly}}
\newcommand{\polylog}{\text{polylog}}

\newcommand\longldots{\makebox[2em][c]{.\hfil.\hfil.}}
\newcommand\longestldots{\makebox[4em][c]{.\hfil.\hfil.}}

\newcommand{\freq}[4]{\mathsf{freq}^{(#1)}(#2;#3;#4)}
\newcommand{\extra}[4]{\mathsf{e}^{(#1)}(#2;#3;#4)}
\newcommand{\Weight}[5]{\mathcal{W}_{#5}^{(#1)}(#2;#3;#4)}
\newcommand{\CWeight}[5]{\overline{\mathcal{W}}_{#5}^{(#1)}(#2;#3;#4)}

\newcommand{\bbN}{\mathbb{N}}
\newcommand{\bbR}{\mathbb{R}}

\newcommand{\bbE}{\mathbb{E}}
\newcommand{\bbP}{\mathbb{P}}

\newcommand{\cD}{\mathcal{D}}
\newcommand{\cE}{\mathcal{E}}
\newcommand{\cF}{\widehat{\mathcal{F}}}
\newcommand{\cG}{\widehat{\mathcal{G}}}

\newcommand{\cL}{\mathcal{L}}

\newcommand{\cQ}{\widehat{\mathcal{Q}}}

\newcommand{\cW}{\mathcal{W}}

\newcommand{\sq}{\normalfont\textsc{q}}
\newcommand{\sd}{\textsc{d}}

\newcommand{\eps}{\varepsilon}

\makeatletter
\newcommand{\customlabel}[2]{%
   \protected@write \@auxout {}{\string \newlabel {#1}{{#2}{\thepage}{#2}{#1}{}} }%
   \hypertarget{#1}{#2}
}
\makeatother

\newcommand{\longdef}[3]{
    \stepcounter{figure}
    \vspace{0.15cm}
    {\small
    \begin{tcolorbox}[breakable, enhanced, colback=halfgreen!10, before upper={\parindent15pt\noindent}]
    \begin{center}
    {\bf \underline{\customlabel{def:#2} #1}}
    \end{center}
    
    #3
    \end{tcolorbox}
    }
}

\newcounter{casenum}

\newcounter{subcasenum}

\newcounter{casenump}

\newcommand{\casep}[2]{
    \ifthenelse{\equal{\value{casenump}}{0}}{
    \vskip.5\baselineskip\par\noindent
    }{}
    {\it Case \arabic{casenump}:} {\it #1}
    \vskip0.1\baselineskip
    \begin{addmargin}[1.5em]{1em}
    #2
    \end{addmargin}
    \addtocounter{casenump}{1}
}

\newcounter{subcasenump}

\begin{document}

\title{Constant Query Local Decoding Against Deletions Is Impossible}
\author{Meghal Gupta \medskip \\UC Berkeley\\ \small{\texttt{meghal@berkeley.edu}}}
\date{\today}

\sloppy
\maketitle
\begin{abstract}
Locally decodable codes (LDC's) are error-correcting codes that allow recovery of individual message indices by accessing only a constant number of codeword indices. For substitution errors, it is evident that LDC's exist -- Hadamard codes are examples of $2$-query LDC's. Research on this front has focused on finding the optimal encoding length for LDC's, for which there is a nearly exponential gap between the best lower bounds and constructions.

Ostrovsky and Paskin-Cherniavsky (ICITS 2015) introduced the notion of local decoding to the insertion and deletion setting. In this context, it is not clear whether constant query LDC's exist at all. Indeed, in contrast to the classical setting, Block et al. conjecture that they do not exist. Blocki et al. (FOCS 2021) make progress towards this conjecture, proving that any potential code must have at least exponential encoding length. 

Our work definitively resolves the conjecture and shows that constant query LDC's do not exist in the insertion/deletion (or even deletion-only) setting. Using a reduction shown by Blocki et al., this also implies that constant query locally correctable codes do not exist in this setting.

% Block et al. (FSTTCS 2020) ask the same question in the insertion/deletion setting: Is there an encoding function $\C: \{0,1\}^n \to \{0,1\}^M$ such that any one bit $x[i]$ of the original message can be recovered with high probability using only a constant number of queries to the corrupted transmission? 
\end{abstract}
\thispagestyle{empty}
\newpage
\tableofcontents
\pagenumbering{roman}
\newpage
\pagenumbering{arabic}

\section{Introduction} \label{sec:intro}

The study of classical error-correcting codes dates back to Shannon and Hamming in the mid-1900’s~\cite{Shannon48, Hamming50}. While the initial focus was on substitution errors, the field expanded to address \emph{synchronization errors} with the seminal work of~\cite{Levenshtein66}. Synchronization errors involve inserting or deleting symbols in a transmitted message rather than substituting symbols. Such errors are prevalent in applications such as text/speech processing~\cite{CoumouS08}, DNA storage technologies~\cite{GabrysYM16, LenzSWY19}, and communication complexity~\cite{BravermanGMO17}.

Within the study of classical (substitution) error-correcting codes, the concept of \emph{locally decodable codes} (LDC's) has garnered considerable attention. In a locally decodable encoding $\C$ of a message $x \in \Sigma^n$ for some alphabet $\Sigma$, a receiver can decode any individual symbol $x[i]$ by making a constant number of queries\footnote{The number of queries is known as the locality, and although it is sometimes useful to explore parameter regimes where the locality is super-constant (for example $\polylog(n)$), in this paper we focus on codes with constant locality.} to $\C(x)$, even in the presence of a small constant fraction of adversarial corruption. It is evident that such codes exist in the classical setting, for example Hadamard codes, and research has focused on finding the optimal encoding length. Despite considerable effort, there remains a nearly exponential gap between the best lower bounds and constructions~\cite{Yekhanin12}.

Recently, the work of Ostrovsky and Paskin-Cherniavsky~\cite{OstrovskyP15} introduced the notion of local decoding to insertion/deletion codes. Unlike in the setting of substitutions, here it is not clear whether constant query LDC's exist at all: For a constant-sized alphabet $\Sigma$, is there an encoding function $\C: \Sigma^n \to \Sigma^M$ such that any one symbol $x[i]$ of the original message can be recovered with high probability using only a constant number of queries to the corrupted transmission?

Block et al.~\cite{BlockBGKZ20} conjecture that in contrast to the substitution setting, constant query insertion/deletion LDC's do not exist. That is, even when the encoding can be arbitrarily long in terms of $n$, constant-query LDC's do not exist in the insertion/deletion setting. Subsequent research by Blocki et al.~\cite{BlockiCGLZZ21} makes progress toward this conjecture, demonstrating that any candidate LDC must have at least exponential length relative to $n$, but it does not definitively rule out their existence. Blocki et al.'s negative result holds even when the adversary can only perform deletions (without insertions).

In this work, we fully resolve Block et al.'s conjecture, establishing that constant-query LDC's of any length do not exist in the insertion/deletion setting. Like \cite{BlockiCGLZZ21}, our result holds when the adversary can only perform deletions. Moreover, due to a reduction shown by \cite{BlockiCGLZZ21}, this implies that locally correctable codes (LCC's) do not exist in the deletion setting. 
%Finally, we show that strong relaxed LDC's (SRLDC's)~\cite{BlockBCKGLZ22, BlockB23} do not exist in the adversarial deletion setting, strengthening the exponential lower bound shown by~\cite{BlockBCKGLZ22}.

\subsection{Our Results}

Our main result is that constant query locally decodable codes of any length do not exist in the deletion setting. We refer the reader to Definition~\ref{def:ldc} for a formal definition of locally decodable codes for the deletion channel.

\begin{restatable}{theorem}{mainthm} \label{thm:main}
    For any $\eps>0$, $k\in \bbN$, and alphabet $\Sigma$, there exists a constant $C:=C(\eps, k, |\Sigma|)$ such that for all $n>C$, there is no $k$-query deletion LDC $\C: \Sigma^n \to \Sigma^M$ for any $M$ that is resilient to $\eps$-fraction of deletions.
\end{restatable}

We conjecture a stronger version of Theorem~\ref{thm:main}, which we hope to see proven or disproven in future work.

\begin{conjecture} \label{conj:main}
    For any $\eps>0$ and $k\in \bbN$, and alphabet $\Sigma$, there exists a constant $C:=C(\eps, |\Sigma|)$ such that for all $n>Ck$, there is no non-adaptive $k$-query deletion LDC $\C: \Sigma^n \to \Sigma^M$ that is resilient to $\eps$-fraction of adversarial deletions.
\end{conjecture}

In other words, we conjecture that any LDC for deletions requires a number of queries directly proportional to the length of the original message to recover a single message symbol. We emphasize that this conjectured lower bound only applies when all queries to a candidate LDC must be specified non-adaptively.
%\footnote{This distinction is not important in the constant query setting of Theorem~\ref{thm:main} because any adaptive query LDC can be transformed into a non-adaptive one with an exponential blow-up in the number of queries.} We mention this distinction to clarify that Conjecture~\ref{conj:main} does not contradict prior results~\cite{OstrovskyP15, BlockBGKZ20} that demonstrate the existence of $\polylog(n)$ query deletion LDC's in the adaptive setting.
When the queries can be submitted adaptively, prior results~\cite{OstrovskyP15, BlockBGKZ20} demonstrate that $\polylog(n)$ queries is achievable.

Theorem~\ref{thm:main} also implies that locally correctable codes do not exist in the deletion setting. The implication follows from a result in \cite{BlockiCGLZZ21}, which states that the existence of LCC's for insertions/deletions would imply LDC's for insertions/deletions. We remark that they present this reduction for adversarial codes allowing both insertions and deletions, but their proof also works in the deletion-only setting.
%Next, we show the analogous result to Theorem~\ref{thm:main} for LCC's, that they do not exist in the oblivious deletion setting. We remark that~\cite{BlockiCGLZZ21} shows a black-box result that the non-existence of deletion LDC's implies the non-existence of deletion LCC's. However, their result is stated only for the adversarial setting, so we prove it directly in the oblivious setting.

\begin{restatable}{corollary}{mainlcc}\label{cor:lcc}
    For any $\eps>0$, $k\in \bbN$, and alphabet $\Sigma$, there exists a constant $C:=C(\eps, k, |\Sigma|)$ such that for all $n>C$, there is no $k$-query deletion LCC $\C: \Sigma^n \to \Sigma^M$ for any $M$ that is resilient to $\eps$-fraction of adversarial deletions.
\end{restatable}

% Finally, we show that strong relaxed LDC's do not exist in the adversarial deletion setting. Relaxed LDC's are strictly weaker than LDC's and allow for some probability of returning a decoding failure answer $\perp$, but a significant proportion of indices must remain decodable. We refer the reader to Definition~\ref{def:srldc} for a formal definition.

% \begin{restatable}{theorem}{mainsrldc}\label{thm:srldc}
%     For any $\eps>0$, $k\in \bbN$, and alphabet $\Sigma$, there exists a constant $C:=C(\eps, k, |\Sigma|)$ such that for all $n>C$, there is no $k$-query deletion SRLDC $\C: \Sigma^n \to \Sigma^M$ for any $M$ that is resilient to $\eps$-fraction of adversarial deletions.
% \end{restatable}

%One can conjecture an analog of Conjecture~\ref{conj:main} for LCC's.

\subsection{Related Works}

\paragraph{Locally decodable codes for insertion/deletions.} Locally decodable codes for insertions and deletions were introduced by Ostrovsky and Paskin-Cherniavsky in~\cite{OstrovskyP15}. They present deletion LDC's of length $\tilde{O}(n)$, resilient against a constant fraction of adversarial deletions with $\polylog(n)$ queries. Cheng, Li, and Zheng~\cite{ChengLZ20} introduce locally decodable codes with randomized encoding and demonstrate improved rate-query tradeoffs in this setting. \cite{BlockBGKZ20} replicate \cite{OstrovskyP15}'s results using different techniques and conjecture the non-existence of deletion LDCs of any length in the constant query regime.

As previously mentioned, Blocki et al.~\cite{BlockiCGLZZ21} make progress on this conjecture, proving that any candidate LDC must have exponential length relative to $n$. More specifically, they eliminate the possibility of $2$-query linear deletion LDCs and establish an exponential lower bound for the size of any $2$-query deletion LDC. In general, they prove a lower bound of $\exp(n^{1/O(k)})$ for a candidate $k$-query LDC. 
%Like ours, their results hold with an oblivious adversary.

~\cite{BlockBCKGLZ22} introduces the relaxed local decoding model (RLDC's) to the insertion/deletion setting, which allows limited decoding failure answers $\perp$. They show lower bounds for constant query ``strong'' RLDC's and constant query constructions for ``weak'' RLDC's. A few other models for insertion/deletion LDC's have been explored, including the secret key setting in~\cite{BlockB21, BlockiCGLZZ21}.  

\paragraph{Insertion/deletion codes.}

The study of codes that correct insertions and deletions was initiated by~\cite{Levenshtein66}. Correcting for insertions and deletions is generally more challenging than substitutions, and our understanding of the field remains limited. Nevertheless, efficient constructions of constant rate codes that correct a constant fraction of insertion/deletion errors have been shown, initially in~\cite{SchulmanZ99}. More recently, the work of~\cite{HaeuplerS17} introduces synchronization strings, which transform substitution error-correcting codes into insertion/deletion error-correcting codes in an efficient black-box manner. The optimal error resilience of constant rate insertion/deletion codes was studied by~\cite{GuruswamiHL22}.

Many other aspects of insertion/deletion codes have been explored. For more details, we refer to surveys on this topic, including~\cite{Mitzenmacher09, HaeuplerS21}.

\paragraph{LDC's and LCC's for substitutions.}

Locally decodable codes and locally correctable codes were initially introduced in the classical setting of substitution errors. The ideas behind LDC's and LCC's arose as early as Reed-Muller codes~\cite{Reed54, Muller54} and were first formalized by Katz and Trevisan~\cite{KatzT00}. They have applications in many areas, including in PCP's~\cite{Babai91}, private information retrieval~\cite{Chor98}, and average case complexity~\cite{Trevisan04}.

The best constructions of LDC's and LCC's for constant queries are barely sub-exponential in the length of the message~\cite{Efremenko09}, while lower bounds are polynomial for $k>3$ queries~\cite{KatzT00, WehnerD05, Woodruff07, Alrabiah23}, except in the case of 2-query LDC's and LCC's~\cite{GoldreichKST02, KerenidisD03} and linear 3-query LCC's~\cite{Kothari23}. The question of whether constant query LDC's with polynomial message length exist remains open. We refer the reader to~\cite{Yekhanin12} for a survey of the work on LDC's.
\section{Technical Overview} \label{sec:overview}

In this section, we'll overview our main result that constant query deletion LDC's do not exist. We'll restrict our attention to the case where the alphabet $\Sigma=\{0,1\}$ as the proof is the same for any constant-sized alphabet.

Before delving into the details, let's set up the problem. Given a sufficiently small $\eps>0$ and a constant number of queries $k$, our goal is to show the existence of a constant $C(\eps,k)$ such that for all $n>C(\eps,k)$, there is no function $\C: \{0,1\}^n \to \{0,1\}^M$ for any value of $M$, which satisfies local decoding for deletions. That is, the adversary has a randomized corruption process $\C(x) \to \widehat{\C(x)}$ using at most $\eps$-fraction corruption with the following guarantee. For any candidate local decoding algorithm that non-adaptively queries $k$ indices of $\widehat{\C(x)}$ to guess $x[i]$, there exists an index $i$ and some $x$ for which it fails with probability greater than $\frac12-\eps$.

Throughout the overview, we use $\widehat{q}$ to denote the index of a query in the corrupted codeword $\widehat{\C(x)}$, while $q$ represents its corresponding index in the original uncorrupted codeword $\C(x)$. For convenience, we'll use $\tau$ to represent a generic constant much smaller than $\eps$ (around $\poly(\eps)$).

\subsection{Reduction to Simulating All Queries with a Fixed Constant-Sized Set}

Our main goal will be to show the following: the adversary can perform deletions in such a way that there is a \emph{constant-sized} set $\cQ$ of lists of queries $(\widehat{\sq_1}\ldots \widehat{\sq_k})$ such that any pair of queries $(\widehat{q_1}\ldots \widehat{q_k})$ can be simulated by a corresponding pair of queries $(\widehat{\sq_1}\ldots\widehat{\sq_k}) \in \cQ$. That is, for most $x$, the distribution over $\{0,1\}^k$ of the output of queries $(\widehat{q_1}\ldots\widehat{q_k})$ is approximately the same (TV distance $<\tau$) as the output of $(\widehat{\sq_1}\ldots\widehat{\sq_k})$ over the randomness of the corruption process.

Pictorially, we want to construct $\cQ$ so that for any list of queries $(\widehat{q_1}\ldots\widehat{q_k})$, the following two processes yield roughly the same distribution of outputs for most $\C(x)$. 
\begin{enumerate} [label=(\arabic*)]
    \item Perform the randomized deletion pattern, then query $(\widehat{q_1}\ldots\widehat{q_k})$ so that the induced indices of the queries in the original codeword are $(q_1\ldots q_k)$, and receive as output the values $\C(x)[q_1]\ldots \C(x)[q_k]$.
    \item Choose the representative query list $(\widehat{\sq_1}\ldots\widehat{\sq_k}) \in \cQ$ that approximates $(\widehat{q_1}\ldots\widehat{q_k})$. Then, perform the randomized deletion pattern, query $(\widehat{\sq_1}\ldots\widehat{\sq_k})$ so that the induced indices of the queries in the original codeword are $(\sq_1\ldots \sq_k)$, and receive as output the values $\C(x)[\sq_1]\ldots \C(x)[\sq_k]$.
\end{enumerate}

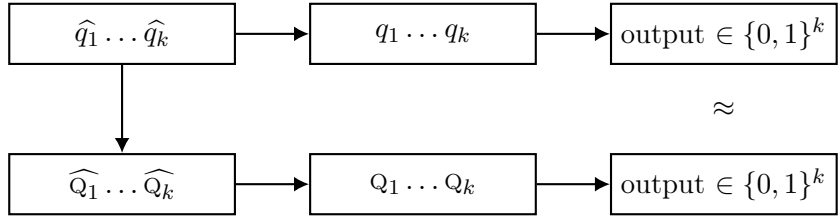
\begin{center}
\begin{tikzpicture}
    \tikzset{every node/.style={minimum width=3cm,minimum height=0.8cm,thick}}
    
    \node[shape=rectangle,draw] (A) at (0,0) {$\widehat{q_1}\ldots\widehat{q_{k}}$};
    \node[shape=rectangle,draw] (B) at (4,0) {$q_1\ldots q_k$};
    \node[shape=rectangle,draw] (C) at (8,0) {output $\in \{0,1\}^k$};

    \node[shape=rectangle,draw] (D) at (0,-2) {$\widehat{\sq_1}\ldots\widehat{\sq_{k}}$};
    \node[shape=rectangle,draw] (E) at (4,-2) {$\sq_1\ldots \sq_k$};
    \node[shape=rectangle,draw] (F) at (8,-2) {output $\in \{0,1\}^k$};

    \node[] (G) at (8,-1) {$\approx$};
    
    \path [->] (A) edge[thick] node {} (B);
    \path [->] (B) edge[thick] node {} (C);
    \path [->] (D) edge[thick] node {} (E);
    \path [->] (E) edge[thick] node {} (F);
    \path [->] (A) edge[thick] node {} (D);
\end{tikzpicture}
\end{center}

%As a preview, one reason we can make such an argument is that two nearby queries, for instance $i$ and $i+1$, have a similar distribution of indices of the original codeword $\C(x)$ they correspond to, provided the adversary's randomized deletion process is sufficiently smooth. Therefore, they produce similar output distributions, so they simulate one another. This type of argument is unique to the deletion setting since it relies on nearby queries having similar outputs post-corruption, which is not true in the substitution setting.

If we can successfully establish this statement, then any local decoding algorithm can be emulated by one that only queries lists in $\cQ$. Consequently, we can conclude that constant query deletion LDCs do not exist, as querying a fixed constant-sized list provides only a constant amount of information about $x$, while the number of choices for the index $i$ is super-constant. A similar information-theoretic argument also shows that constant query deletion LCC's do not exist.

For the remainder of this overview, we will show how to design the adversary's corruption process and find this constant-sized $\cQ$. We will begin with the $2$-query case which illustrates many of the ideas, and then we will describe how to modify the argument for the general $k$-query case.

\subsection{The $\mathbf{2}$-Query Case}

\subsubsection{The Corruption Process.} We'll begin by defining the corruption process used by the adversary. This is the same process used in the work~\cite{BlockiCGLZZ21}. The adversary performs two types of corruptions to $\C(x)$.\footnote{We remark that the formal deletion process described in Section~\ref{sec:del-pattern} even when restricted to the $k=2$ case is different and more complex, but that this simpler process used in the overview permits the same analysis specifically for $k=2$.}
\begin{itemize}
    \item \textbf{Type 1:} Choose $\eps_1\in [0,\frac{\eps}{2}]$ uniformly at random and delete  the first $\eps_1M$ bits of $\C(x)$.
    \item \textbf{Type 2:} Choose $\eps_2\in [0,\frac{\eps}{2}]$ uniformly at random and delete every $\frac{1}{\eps_2}$'th remaining bit.
\end{itemize}

\subsubsection{Constructing $\cQ^{(\ell)}$ for Each Layer $\ell$.}

We'll first attempt to construct $\cQ$ so that it satisfies an even stronger guarantee than that the output in $\{0,1\}^2$ of any pair $(\widehat{q_1},\widehat{q_2})$ can be approximated by some pair $(\widehat{\sq_1},\widehat{\sq_2}) \in \cQ$. We'll have $\cQ$ contain, for every possible pair of queries $(\widehat{q_1},\widehat{q_2})$, a corresponding pair $(\widehat{\sq_1},\widehat{\sq_2}) \in \cQ$ such that the distribution of $(q_1,q_2)$ is similar to $(\sq_1, \sq_2)$. This is essentially what~\cite{BlockiCGLZZ21} do, and the argument we outline below is a reformulation of theirs.

We will show that if $\cQ$ contains $\widehat{\sq_1} \in \big[\widehat{q_1}\pm \tau M\big]$ and $\widehat{\sq_2} \in \big[\widehat{q_2}\pm \tau (\widehat{q_2}-\widehat{q_1})\big]$, then the statement holds true. With just Type 1 deletions, shifting a query pair $(\widehat{q_1},\widehat{q_2})$ by $(\tau M,\tau M)$ maintains the distribution of the $q_1$ and $q_2$. This is because the small difference in the indices queried is counteracted by the unpredictability of $\eps_1$. Similarly, Type 2 deletions counteract a small change in the distance between $\widehat{q_1}$ and $\widehat{q_2}$, so they can be shifted by $(0,(\widehat{q_2}-\widehat{q_1})\tau)$ without changing the output by much.

%It will turn out that we can find a deletion process such that the ``unpredictability'' (how large the range of plausible values is) of $q_1$ (the induced index of $\widehat{q_1}$ into the original codeword) proportional to $M$, so $\sq_1$ will have a similar distribution since $\widehat{\sq_1}$ differs from $\widehat{q_1}$ by at most $\tau M$. Since the indices $q_1$ and $q_2$ are related, the unpredictability of $q_2$ will only be proportional to $(\widehat{q_2}-\widehat{q_1})$ rather than $M$.

Notice that a "random" deletion process, where every index is deleted with probability $\eps/2$, doesn't work here, because the unpredictability of the index $q_1$ is proportional only to $\sqrt{q_1}$ and not to $M$. Even if the deletion fraction is chosen randomly instead of being $\eps/2$, the unpredictability of $q_2$ conditioned on $q_1$ is only proportional to $\sqrt{\widehat{q_2}-\widehat{q_1}}$ and not to $(\widehat{q_2}-\widehat{q_1})$.

%Instead, the adversary will use the following deletion process. \mnote{put the deletion process first, cite them}

Now, we want to construct $\cQ$ so that it has a pair $\widehat{\sq_1} \in \big[\widehat{q_1}\pm \tau M\big]$ and $\widehat{\sq_2} \in \big[\widehat{q_2}\pm \tau (\widehat{q_2}-\widehat{q_1})\big]$ for every $(\widehat{q_1},\widehat{q_2})$. For a fixed ``layer'' $\ell$ of queries -- that is queries satisfying $\ell:=\lfloor \log(\widehat{q_2}-\widehat{q_1}) \rfloor$, only a constant number of queries need to be added to $\cQ$ to approximate all queries in this layer. This is because only one option for $\widehat{q_1}$ is needed in each $\tau M$-length range, and only one option for $\widehat{q_2}-\widehat{q_1}$ in each $\tau \ell$-length range. To be explicit, any such pair of queries can be approximated by one of the form $(\tau Mi, \tau Mi + \tau 2^aj)$ for $i\in [\tau^{-1}], j\in [\tau^{-1}]$. We'll add all queries of this form into the set $\cQ^{(\ell)}$.

%This is because the distribution of corresponding indices in the original message is very similar since they are counteracted by the random choice of $\eps_1$. Similarly, Type 2 deletions counteract a small change in the distance between $\widehat{q_1}$ and $\widehat{q_2}$. More precisely, $(\widehat{q_1},\widehat{q_2})$ has a similar output distribution as $(\widehat{q_1},\widehat{q_2} + (\widehat{q_2}-\widehat{q_1})\tau)$. Combining the effects of both types of deletions, the query $(\widehat{q_1},\widehat{q_2})$ induces the same output as the query $(\widehat{q_1}+\tau_1 M,\widehat{q_2})+\tau_1 M + \tau_2 (\widehat{q_2}-\widehat{q_1})$ for $\tau_1,\tau_2 \ll \eps$.

%We categorize the queries into layers based on the approximate value of $(\widehat{q_2}-\widehat{q_1})$. Formally, we define $\ell:=\lfloor \log(\widehat{q_2}-\widehat{q_1}) \rfloor$. The argument above shows that for a fixed layer $\ell$ (that is, all queries satisfying $\ell=\lfloor \log(\widehat{q_2}-\widehat{q_1}) \rfloor$), there are constantly many queries that model all possible queries. To be explicit, any such query can be approximated by one of the form $(\tau Mi, \tau Mi + \tau 2^aj)$ for $i\in [\tau^{-1}], j\in [\tau^{-1}]$. We'll add all queries of this form into the set $\cQ^{(\ell)}$.

If there were only a constant number of layers, that is if $\log{M}$ were a constant, then we would be done, as all queries can be approximated by queries in $\bigcup_{\ell\in [\log{M}]} \cQ^{(\ell)}$. Unfortunately, $\log{M}$ need not be a constant, and in fact will certainly depend on $n$. The main goal of the rest of our argument is to show that we only need to include $\cQ^{(\ell)}$ for a \emph{constant number} of layers $\ell$, and these alone will approximate \emph{all} queries. This constant-sized set of layers will be represented as $\cL$.

We remark that this argument suffices for the bound in~\cite{BlockiCGLZZ21}. It is actually enough that $|\cQ|=o(n)$ to achieve the information-theoretic contradiction in the previous section, so they establish a contradiction whenever $\log(M) = o(n)$, which yields their exponential lower bound.

\subsubsection{A Constant Number of Layers $\ell$ Suffice.}

If we let $\cQ=\bigcup_{\ell\in [\log{M}]} \cQ^{(\ell)}$, it is not only true that $(\widehat{q_1},\widehat{q_2})$ shares the same distribution of outputs in $\{0,1\}^2$ as their representative queries in $\cQ$, but they also correspond to a similar distribution on indices of the uncorrupted codeword $\C(x)$. As explained earlier, this is stronger than necessary. We'll identify layers $\ell$ that cen be excluded because they induce the same output distribution as other layers, even if not the same distribution on indices of the uncorrupted codeword $\C(x)$.

\paragraph{An illustrative example.} Let's with an example. Look at the following string:
\[
    \overbrace{\underbrace{00\ldots0}_{\sqrt{n}\text{ 0's}}~\underbrace{11\ldots1}_{\sqrt{n}\text{ 1's}} \longldots \underbrace{00\ldots0}_{\sqrt{n}\text{ 0's}}~\underbrace{11\ldots1}_{\sqrt{n}\text{ 1's}}}
    ^{\sqrt{n}\text{ blocks}}
\]

In layers where $\widehat{q_2}-\widehat{q_1} \ll \sqrt{n}$, any two queries $(\widehat{q_1},\widehat{q_2})$ have a very high probability of being identical. This is because the random shift from the Type 1 deletions mean that $q_1$ (the induced index of $\widehat{q_1}$ in the original codeword) is fairly uniform within the block it falls into, making it likely that $q_2$ is in the same block. Thus, slightly changing the gap between the queries does not significantly affect the output distribution, and hence, the set $\cQ^{(1)}$ accounts for all $\cQ^{(\ell)}$ with $2^\ell \ll \sqrt{n}$.

A similar argument works for queries where $\widehat{q_2}-\widehat{q_1} \gg \sqrt{n}$. In this case, $q_1$ and $q_2$ will be in different blocks with many blocks between them. In particular, Type 2 errors make it approximately equally likely for there to be an odd or even number of blocks between $\widehat{q_1}$ and $\widehat{q_2}$, and so their values are uncorrelated. Since this is true for all layers where $2^\ell\gg \sqrt{n}$, any single $\cQ^{(\ell)}$ correctly models all the queries in these layers.

Hence, to model all queries, $\cL$ would need only to include $\ell=1$, one value of $\ell$ much larger than $\log\sqrt{n}$, and a few values of $\ell$ near $\log{\sqrt{n}}$. A more detailed analysis shows that only constantly many layers of the third type are necessary. This results in a constant sized union $\cQ$.

Surprisingly, this phenomenon is not limited to the specific example shown; most strings exhibit a similar behavior. For a general string, consider the fraction of 0's and 1's in each block of size $2^\alpha$ and $2^\beta$ for some $\beta \gg \alpha$.
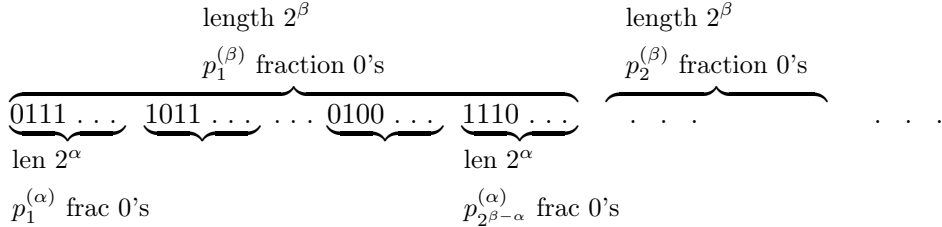
\[
    \overbrace{
    \underbrace{0111\longldots}_
    {\mathclap{\small{\begin{aligned}&~~~\text{len } 2^\alpha \\ &~~~p^{(\alpha)}_1 \text{ frac 0's}\end{aligned}}}}
    ~\underbrace{1011\longldots}_{} \longldots \underbrace{0100\longldots}_{}~\underbrace{1110\longldots}_
    {\mathclap{\small{\begin{aligned}&~~~~~\text{len } 2^\alpha \\ &~~~~~p^{(\alpha)}_{2^{\beta-\alpha}} \text{ frac 0's}\end{aligned}}}}
    }
    ^{\small{\begin{aligned}&\text{length } 2^\beta \\ &p^{(\beta)}_1 \text{ fraction 0's}\end{aligned}}}
    ~~\overbrace{\vphantom{0}\longestldots~~~~~~~~~~~}
    ^{\small{\begin{aligned}&\text{length } 2^\beta \\ &p^{(\beta)}_2 \text{ fraction 0's}\end{aligned}}}
    ~~\longestldots
\]

The main claim is as follows. Assume that most blocks of size $2^\alpha$ within the same $2^\beta$ chunk $i$ have approximately the same fraction of 0's (i.e. $p^{(\alpha)}_{2^\beta i+j}$ is similar for all $j\in [2^{\beta-\alpha}]$, which also implies all $p^{(\alpha)}_{2^\beta i+j}$ are approximately equal to the mean $p^{(\beta)}_i$). Then only a single layer $\ell \in [\alpha +\tau^{-1},\beta -\tau^{-1}]$ needs to be included in $\cL$ (within this layer, it holds that $2^\alpha \ll \widehat{q_2}-\widehat{q_1} \ll 2^\beta$).

To illustrate this, consider a pair of queries $(\widehat{q_1},\widehat{q_2})$. Conditioned on $q_1$ falling within a specific size $2^\beta$ chunk $i$ (and $q_2$ will likely fall in the same chunk since $\widehat{q_2}-\widehat{q_1} \ll 2^\beta$), we'll show that the distribution of the output of queries $(\widehat{q_1},\widehat{q_2})$ is $\text{Bern}(p^{(\beta)}_i)^2$.\footnote{We use $\text{Bern}(p)$ for the distribution on $\{0,1\}$ with probability $p$ of being $0$ and probability $1-p$ of being $1$.} Because $\widehat{q_2}-\widehat{q_1} \gg 2^\alpha$, the Type 2 errors cause enough unpredictability in $q_2-q_1$ that both $q_1$ and $q_2$ can be considered essentially independent and uniformly distributed within the size $2^\alpha$ chunk that each falls into (even though which size $2^a$ each query falls into can be correlated). We've assumed that each $2^\alpha$ block has distribution $\text{Bern}(p^{(\beta)}_i)$ for a query made randomly to it, so the joint distribution on $(\widehat{q_1},\widehat{q_2})$ is $[\text{Bern}(p^{(\beta)}_i)]^2$. 

Consequently, the output is essentially the same for all $\ell \in [\alpha +\tau^{-1},\beta -\tau^{-1}]$ within the $2^\beta$-sized block of the original codeword that the queries fall into. In particular, any query $(\widehat{q_1}, \widehat{q_1}+D)$ exhibits almost an identical same output distribution as long as $\alpha +\tau^{-1}< \log D < \beta -\tau^{-1}$. Therefore, it suffices to represent all of the queries in a single one of the layers $\ell$ rather than all of them, which is adequately done by one $\cQ^{(\ell)}$.

\paragraph{A structure lemma.} It may seem like a strong assumption that most blocks of size $2^\alpha$ within the same $2^\beta$ chunk have approximately the same fraction of 0's, but it is indeed the case that any string has only constantly many layers at which this fraction changes. Specifically, we'll identify a list $\alpha_1 < \ldots < \alpha_{\tau^{-1}}$ such that most blocks of size $2^{\alpha_i}$ within the same $2^{\alpha_{i+1}}$ chunk have approximately the same fraction of 0's. Then, only one layer between $\alpha_i+1$ and $\alpha_{i+1}$ needs to be included in $\cL$ since they all then $\cQ^{(\ell)}$ induce the same distribution of outputs.\footnote{This is not quite true; actually a constant number of extra layers right around $\alpha_i+1$ and $\alpha_{i+1}$ need to be included.}

Consider the variance of $p^{(\alpha)}$ within a size $2^\beta$ chunk (for simplicity of notation let's focus on the first size $2^\beta$ chunk).
\[
    \mathop{\text{Var}}_{i\sim [2^{\beta-\alpha}]}\Big(p^{(\alpha)}_i\Big) = \mathop{\bbE}_{j\sim [2^{\beta-\alpha}]} \Big[(p^{(\alpha)}_j)^2 \Big] - \Big(p^{(\beta)}_1\Big)^2.
\]
We want to characterize the ranges in which this variance is large. For any $\theta \in [\log{M}]$, define the weight $\cW_\theta:= \mathop{\bbE}_{i\sim [M/2^\theta]} \big[p^{(\theta)}_i \big]$. Then, averaging the above variance (of $2^\alpha$ sized blocks) over all the sized $2^\beta$ chunks gives
\begin{align*}
    \mathop{\bbE}_{i\in [M/2^\beta]} \Bigg[ \mathop{\text{Var}}_{j\sim [2^\beta i+1, 2^\beta i + 2^{\beta-\alpha}]}\Big(p^{(\alpha)}_j\Big) \Bigg] &= \mathop{\bbE}_{i\in [M/2^\beta]} \Bigg[ \mathop{\bbE}_{j\sim [2^\beta i+1, 2^\beta i + 2^{\beta-\alpha}]} \Big[(p^{(\alpha)}_j)^2 \Big] - (p^{(\beta)}_i)^2 \Bigg]. \\
    &=\cW_\alpha-\cW_\beta
\end{align*}

In simpler terms, we have found a weight function on layers $\cW_\theta$ such that the variance of the fraction of $0$'s in a sized $2^\alpha$ block in a $2^\beta$ chunk is $W_\alpha-W_\beta$. Moreover, this weight $\cW_\theta$ decreases as $\theta$ increases and always falls between $0$ and $1$. Therefore, the string can be partitioned into a constant length list $\alpha_1\ldots \alpha_{\tau^{-1}}$ such that for all $i$, $\cW_{\alpha_{i+1}}-\cW_{\alpha_i+1} \ll \eps$ and so only a single layer in each interval needs to be included.

\paragraph{Reconciling with $\C$.} Recall that $\cQ$ needs to be a function of the code $\C$, not each individual string the sender could send. To deal with this, one can instead use $\cW_\theta(\C):=\bbE_{z\in \C} \big[ \cW_\theta(z) \big]$, and find the relevant $\alpha_1\ldots \alpha_{\tau^{-1}}$ according to this weighting. Then, on most strings $z\in \text{Im}(\C)$, layers between adjacent elements of the sequence will induce the same output distribution, which is enough for the argument to work.

\subsection{The \textit{k}-Query Case}

Next, we'll extend this argument to the $k$-query case.

\subsubsection{Naive Attempt to Extend 2-Query Argument.} 
Ideally, the same corruption pattern and a similar analysis would still apply. We had identified a sequence of critical transition layers $\alpha_1\ldots \alpha_{\tau^{-1}}$ using the weight function defined above, and based on this, pinpointed a constant-sized list of layers $\cL$ that adequately approximated all pairs of queries. We can attempt to extend this by defining $\cQ^{(\ell_1\ldots \ell_{k-1})}$ for $\ell_1\ldots\ell_{k-1} \in \cL$ as follows: choose $j, i_1\ldots i_{k-1} \in [\tau^{-1}]$ and include each tuple of the form $\big(j\tau M, j\tau M+ i_12^{\ell_1},\ldots, j\tau M + i_12^{\ell_1}+\ldots+i_{k-1}2^{\ell_{k-1}}\big)$ in $\cQ^{(\ell_1\ldots \ell_{k-1})}$.

To carry through the $2$-query analysis, we would need to show the following two claims:
\begin{enumerate} [label=(\arabic*)]
    \item\label{itm:fix1} Any $(\widehat{q_1},\ldots, \widehat{q_k})$ satisfying $\widehat{q_{i+1}}-\widehat{q_i}\approx 2^{\ell_i}$ induces the same distribution on $(q_1,\ldots, q_k)$ as some list in $\cQ^{(\ell_1\ldots \ell_{k-1})}$.
    \item\label{itm:fix2} Any successive difference between queries $(\widehat{q_{i+1}}-\widehat{q_i})$ that does not correspond to a layer in $\cL$ is suitably approximated by a layer in $\cL$.
\end{enumerate}
Unfortunately, neither of these claims is true as stated. 

\subsubsection{Addressing the First Issue.} 
To see this issue with \ref{itm:fix1}, consider the following triples of queries for some fixed $q,d$ where $d\approx 2^\ell$ for some layer $\ell \in \cL$: $(\widehat{q_1}:=q,\widehat{q_2}:=q+d,\widehat{q_3}:=q+2d)$ and $(\widehat{q_1'}:=q,\widehat{q_2'}:=q+d,\widehat{q_3'}:=q+2d+10)$. We want the distributions of $(q_1,q_2,q_3)$ and $(q_1',q_2',q_3')$ to be similar. However, since $\widehat{q_3}-\widehat{q_2}=\widehat{q_2}-\widehat{q_1}$, any induced indices $(q_1,q_2,q_3)$ always take on the form $(Q, Q+D, Q+2D)$ for some $Q, D$ since Type 2 deletions delete the same fraction of each equal-length interval. By similar logic $(q_1',q_2',q_3')$ takes on the form $(Q,Q+D,>Q+2D+5)$. Consequently, not only are the induced distributions of $(q_1,q_2,q_3)$ and $(q_1',q_2',q_3')$ different, but they have no overlap, as the triples take on different forms.

The key issue here is that every interval of length $D$ in a codeword has the same number of bits deleted, while the desired behavior is to have a random, uncorrelated number of bits deleted. To achieve this, we'll modify the deletion process as follows: for every layer $\ell \in \cL$, in the $i$'th chunk of length $2^\ell$, choose a random fraction $\eps^{(\ell)}_i$ and delete every $\frac{1}{\eps^{(\ell)}_i}$'th bit. Since there are only a constant number of layers in $\cL$, this still results in a constant fraction of corruption overall. Now, each interval of length $\approx 2^\ell$ not only has a random fraction of bits deleted, but these fractions are uncorrelated. Consequently, the example from the previous paragraph would induce the same distribution on $(q_1,q_2,q_3)$ and $(q_1',q_2',q_3')$ as originally intended.

\subsubsection{Addressing the Second Issue.} 
The issue with~\ref{itm:fix2} is that $\cL$ doesn't actually encompass all the important layers. We will show that there are important transition layers that are not accounted for by the weight function $\cW$ defined previously. In the final argument, the number of transition layers will still be constant, but we need a new weight function to identify additional important layers.

\paragraph{An illustrative example.} Let's begin with an example that demonstrates why $\cW$ is inadequate.

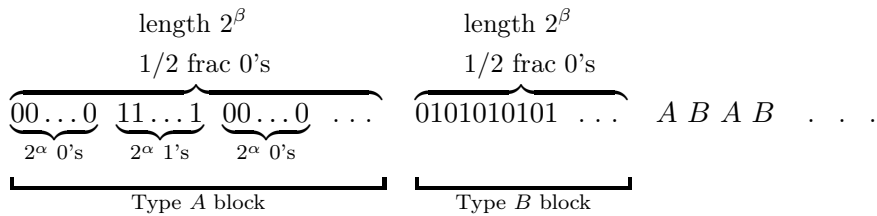
\[
    \underbracket{\overbrace{\underbrace{00\ldots0}_{2^\alpha\text{ 0's}}
    ~\underbrace{11\ldots1}_{2^\alpha\text{ 1's}}
    ~\underbrace{00\ldots0}_{2^\alpha\text{ 0's}}
    ~\longldots}
    ^{\small{\begin{aligned}
        &~~\text{length }2^\beta \\
        &~~1/2\text{ frac 0's}
    \end{aligned}}}
    \vphantom{\underbrace{1}_{\frac12}}}
    _{\text{Type $A$ block}}
    ~~
    \underbracket{\overbrace{0101010101~\longldots}
    ^{\small{\begin{aligned}
        &~~\text{length }2^\beta \\
        &~~1/2\text{ frac 0's}
    \end{aligned}}}
    \vphantom{\underbrace{1}_{\frac12}}}
    _{\text{Type $B$ block}}
    ~~A~B~A~B~\longestldots
\]

In this example, the string consists of $M/2^\beta$ alternating Type $A$ and Type $B$ blocks. At each scale $\theta \in [\log M]$, let's examine the weight $\cW_\theta:= \mathop{\bbE}_{i\sim [M/2^\theta]} \big[p^{(\theta)}_i \big]$. For all $\theta>\alpha$, this quantity is $1/4$, so if $\cW$ were sufficient, differences between queries that are $\gg 2^\alpha$ would need to behave essentially the same. However, it is evident that the structure of the string chanes around the scale $2^\beta$. Two queries weren't sufficient to detect this change, but it turns out that four queries are.

To be concrete, choose $D_1$ and $D_2$ satisfying $2^\alpha \ll D_1 \ll 2^\beta$ and $D_2\gg 2^\beta$. Make the following two sets of queries: $(0, 1, D_1, D_1+1)$ and $(0, 1, D_2, D_2+1)$. The aim is for the distribution on outputs to be the same since $D_1$ and $D_2$ correspond to supposedly indistinguishable layers. However, when a pair of queries differing by $1$ falls in a type $A$ block, the outputs are typically the same, and when such a pair falls in a type $B$ block, the outputs are typically different. In the query $(0, 1, D_1, D_1+1)$, the first pair is likely in the same length $2^\beta$ block as the second pair (so whether the first pair are the same is correlated to whether the second pair are the same). However, in $(0, 1, D_2, D_2+1)$, the first and second pair are likely in different, uncorrelated blocks. Therefore, the distribution of outputs in $\{0,1\}^4$ for the queries $(0, 1, D_1, D_1+1)$ and $(0, 1, D_2, D_2+1)$ are different.

\paragraph{Defining more weight functions.} In this way, layer $\beta$ is a transition point where smaller layers look substantially different than larger ones, and it's not accounted for by $\cW$. Our new weighting needs to say something about what happens when you make \emph{two} queries to this region rather than just considering the fraction of 0's and 1's. Given any pair of queries $(\widehat{q_1},\widehat{q_2})$, define the weight:
\[
    \cW^{(00)}_\theta(\widehat{q_1},\widehat{q_2}) := \mathop{\bbE}
    _{\scriptsize{\begin{aligned}&i\sim [M/2^\theta] \\ &c\sim [2^\theta]\end{aligned}}} 
    \Big[\big(\text{$\Pr$ $(\widehat{q_1}+c,\widehat{q_2}+c)$ outputs $00$} \big)^2\Big].
\]
We similarly define $\cW^{(01)}$, $\cW^{(10)}$, and $\cW^{(11)}$, and if the alphabet $\Sigma$ were larger than binary, we would define $\cW^{(\sigma_1\sigma_2)}$ for all $\sigma_1,\sigma_2 \in \Sigma$, but we'll focus the discussion on $\cW^{(00)}$. In our earlier example, observe that layers smaller than $\beta$ have significantly higher weight than layers larger than $\beta$ for $\widehat{q_1}=0$ and $\widehat{q_2}=1$, so this successfully addresses the issue. Furthermore, for a fixed $(\widehat{q_1},\widehat{q_2})$, this weight is also (approximately) decreasing with respect to $\theta$, so it has a constant number of transition points. Consequently, we can use the constant-sized set of transition points for $\cW^{(00)}_\theta(\widehat{q_1},\widehat{q_2})$ and correspondingly add a constant number more layers into $\cL$ for each $(\widehat{q_1},\widehat{q_2})$. As it turns out, this new $\cL$ successfully accounts for all the relevant layers, at least when $k=3$. 

One issue remains: there are super-constantly many pairs $(\widehat{q_1},\widehat{q_2})$, but we can only add a constant number of layers to $\cL$. This is easily addressed by utilizing only a representative set of $(\widehat{q_1},\widehat{q_2})$ whose elements approximate all possible queries that could be made. This is exactly what the set $\cQ$ calculated for the $2$-query case represents (so we only need to use $\cW^{(00)}_\theta(\widehat{q_1},\widehat{q_2})$ for pairs of queries in the original $\cQ$).

In general, for each of the $k$ queries, we define an additional recursive layer of weightings. In this manner, we construct the set $\cL$ of relevant layers based on those relevant to any weighting, and include every corresponding $\cQ^{(\ell_1\ldots \ell_{k-1})}$. We refer the reader to Section~\ref{sec:kquery} for a comprehensive analysis. It is worth noting that the formal analysis in Section~\ref{sec:kquery} does not define $\cQ$ in in the same way, but the ideas remain the same.
\section{Preliminaries}

\subsection{Model Definitions}

We'll start with the definition of locally decodable codes (LDC's) that are resilient to $\eps$-fraction of adversarial deletions.

\begin{definition} [Deletion LDC]  \label{def:ldc}
    For a fixed value of $n$ and alphabet $\Sigma$, a candidate deletion LDC is an encoding $\C : \Sigma^n \to \Sigma^M$. We say that $\C$ is resilient to $\eps$-fraction of adversarial deletions as a $q$-query deletion LDC (for shorthand, just $k$-query $\eps$-resilient LDC) if the following holds. 

    For all $i \in [n]$, there exists an algorithm that makes $k$ queries\footnote{Arbitrary indices of $\widehat{\C(x)}$ can be queried, and if the index is out of bounds, the query will return $0$.} to a codeword $\C(x)$ adversarially\footnote{The adversary knows the encoding $\C$ and $x$ but not the decoder's algorithm.} corrupted by an adversary using less than $\eps M$ deletions and with probability at least $\frac12+\eps$ correctly guesses $x[i]$ for all values of $x \in \Sigma^n$.
\end{definition}

Next, we'll define locally decodable codes (LCC's) that are resilient to $\eps$-fraction of adversarial deletions.

\begin{definition} [Deletion LCC]  \label{def:lcc}
    For a fixed value of $n$ and alphabet $\Sigma$, a candidate deletion LCC is an encoding $\C : \Sigma^n \to \Sigma^M$. We say that $\C$ is resilient to $\eps$-fraction of adversarial deletions as a $q$-query deletion LCC (for shorthand, just $k$-query $\eps$-resilient LCC) if the following holds. 

    For all $i \in [M]$, there exists an algorithm that makes $k$ queries to a codeword $\C(x)$ adversarially corrupted by an adversary using less than $\eps M$ deletions and with probability at least $\frac12+\eps$ correctly guesses $\C(x)[i]$ for all values of $x \in \Sigma^n$.
\end{definition}

% Finally, we'll define strong relaxed locally decodables (SRLDC's) that are resilient to $\eps$-fraction of adversarial deletions. This model was introduced to the deletion setting by \cite{BlockBCKGLZ22}, who showed a nearly exponential lower bound against such codes. Our definition omits an additional requirement in their notion that the decoder must guarantee success whenever there are no corruptions. We define this in the adversarial model rather than oblivious to better parallel the model in \cite{BlockBCKGLZ22} which doesn't have a clear oblivious analog.

% \begin{definition} [Deletion SRLDC]  \label{def:srldc}
%     For a fixed value of $n$ and alphabet $\Sigma$, a candidate deletion strong LDC is an encoding $\C : \Sigma^n \to \Sigma^M$. We say that $\C$ is resilient to $\eps$-fraction of oblivious deletions as a $q$-query deletion SRLDC (for shorthand, just $k$-query $\eps$-resilient SRLDC) if the following holds. 

%     For all $i \in [n]$ there exists an algorithm that makes $k$ queries to $y\in \Sigma^M$ and outputs either $0,1$, or $\perp$ such that for all $y$ that can be attained by at most $\eps M$ deletions the following hold:
%     \begin{enumerate} [label=(\arabic*)]
%         \item With probability at least $\frac12+\eps$ the algorithm correctly guesses $x[i]$ or $\perp$ for all values of $x \in \Sigma^n$.
%         \item There exists a subset $I_y$ of indices of size $\geq \eps n$ such that on every index $i\in I$, the algorithm outputs $x[i]$ with probability at least $\frac12+\eps$.
%     \end{enumerate}
% \end{definition}

\subsection{Notation and Definitions}

\paragraph{Notation.}
\begin{itemize}
    \item The set $[n]$ denotes the set $\{1\ldots n\}$.
    \item Logarithms are in base $2$ unless otherwise specified.
    \item In general, messages will be denoted by $x$, encoded messages by $\C(x)$, and corrupted encoded messages by $\widehat{\C(x)}$.
    \item By default, if an algorithm makes queries $\widehat{q_1}\ldots \widehat{q_r}$, the queries will be non-adaptive and in increasing order (so $\widehat{q_r}>\dots > \widehat{q_1}$). The induced indices of $\widehat{q_1}\ldots \widehat{q_r}$ in the uncorrupted codeword are $q_1\ldots q_r$.
\end{itemize}

First, we define a (randomized) deletion pattern that the adversary can perform on a string with indices $1\ldots n$.

\begin{definition}[Deletion Pattern]
    A deletion pattern $\del$ on the indices $1\ldots \nu$ is a probability distribution over subsets of $1\ldots n$. Usually, in a string of length $\nu$, the bits at the indices corresponding to the subset will be deleted.
\end{definition}

For the context of our next definition, consider the situation after Bob has received a corrupted codeword and queries the indices $\widehat{q_1} \ldots \widehat{q_k}$. 

\begin{definition}[$\overline\cD(S)$]
    For a measurable space $S$, the family of probability distributions over $S$ is denoted by $\overline\cD(S)$.
\end{definition}

\begin{definition}[$\cD(\widehat{q_1},\ldots,\widehat{q_k}; \del)$]
For a deletion pattern $\del$ on strings of length $n$, given that Bob has queried $\widehat{q_1},\ldots,\widehat{q_k}$, let the distribution of $q_1,\ldots q_k$ (indices of the original string) be denoted $\cD:=\cD(\widehat{q_1},\ldots,\widehat{q_k}; \del)$.

For a distribution $\cD$ and string $z$ of length $n$, the distribution $z[\cD]$ denotes the distribution of $\Sigma^*$ obtained by indexing the tuple $\cD$ into the string $z$. If a queried index in $\cD$ are out of bounds for $z$ (i.e. not in the range $1\ldots n$), then instead return $\err$ (so the distribution also supports some error rather than only valid outputs in $\Sigma^r$.
\end{definition}

%We let the norm $\V D_1,D_2 \V_\TV$ denote the total variation distance between two probability distributions $D_1,D_2$ over a measurable space $T$. Next, we define the total variation distance between two functions from a finite set to a family of distributions. 

% \begin{definition} [TV Distance of Functions]
%     Let $D$ be the space of distributions over a measurable space $T$ and $S$ be a finite set. Given two functions $f_1,f_2: S \to D$, we define the total variation distance between $f_1,f_2$ as follows (where $U(S)$ denotes the uniform distribution over $S$):
%     \[
%         \V f_1,f_2 \V_\TV := \mathop{\bbE}_{s \sim U(S)} \Big[ \V f_1(s),f_2(s) \V_\TV \Big]
%     \]
% \end{definition}

Next, we'll define an easy way to view a string $z$ as a concatenation of chunks.

\begin{definition}[$C_{i,\nu}(z)$]
    Consider a string $z$ whose length is a multiple of $\nu$ (specifically, length $\alpha\nu$). For $i\in [\alpha]$, we define $C_{i,\nu}(z)$ to be the $i$'th length $\nu$ chunk of length of the string $z$. 
\end{definition}

% [old]
% Next, we'll define an easy way to view a string $z$ as chunks of length $2^t$.

% \begin{definition}[$C_t(z)$]
%     We define $C_t(z)$ to be the partition of $z$ into chunks of length $2^t$. More precisely, $z$ is partitioned into $2^{m-t}$ chunks of size $2^t$ labeled $C_t(z)[1]\ldots C_t(z)[2^{m-t}]$ (essentially, $C_t$ changes the indexing of $z$ from blocks of length $1$ to blocks of length $2^t$).
% \end{definition}
\section{Impossibility of \textit{k}-Query Deletion LDC's} \label{sec:kquery}

In this section, we establish the main theorem, that constant query locally decodable codes for deletion errors do not exist. 

\mainthm

Moreover, we establish that constant query LCC's do not exist.

\mainlcc

%\mainsrldc

We'll show these theorems by considering any candidate locally decodable code $\C$ and describing a corruption pattern for which some index $i$ cannot be recovered with only $k$ queries. Throughout this section, $\eps$, $k$ and $n$ are fixed, and we'll let $\kappa$ be the smallest power of $2$ greater than $\eps^{-k\cdot |\Sigma|}$. Let $\C$ denote the candidate LDC (encoding strings of length $n$ to length $M=\kappa^m$), $x$ denotes a message being sent, and $z$ typically is an arbitrary string of length $M$.

We remark that the assumption that $M$ is a power of $\kappa$ is almost without loss of generality. This is because any encoding is effectively padded with infinitely many $0$'s since out-of-bounds queries return $0$. To make the code's length a power of $\kappa$, simply increase the length of the code until the next power of $2$ larger than $M$ which at most doubles the total length. Then, any adversarial attack on the new padded encoding also works on the original (since deleting indices larger than $M$ doesn't change the string) and requires (at most) the same number of deletions, but since $M$ increased in length by a factor of up to $\kappa$, we will require error resilience $\eps\kappa^{-1}$ instead of $\eps$.

\subsection{Definitions and Useful Lemmas}

We begin with a few useful definitions and simple facts. Throughout, the length of an arbitrary string $\nu$ will be assumed to be a power of $\kappa$ (so that it divides $M$).

\subsubsection{Frequencies}

When discussing lists of queries, we'll usually be interested in the list of differences between successive queries rather than the queries themselves. In particular, for any list of queries $\widehat{q_1}\ldots \widehat{q_r}$, we'll define the list of differences $\widehat{d_1}\ldots \widehat{d_{r-1}}$ as $d_i:=\widehat{q_{i+1}} - \widehat{q_i}$. We'll often state lemmas about the lists of query differences rather than the queries themselves.

Our first definition describes the probability of the output being a specific pattern $b$ when queries with differences $\widehat{d_1}\ldots \widehat{d_{r-1}}$ are made to a string $s$ on which some probabilistic deletion pattern is performed. The queries will be randomly shifted in the string subject to the successive differences being $\widehat{d_1}\ldots \widehat{d_{r-1}}$.

\begin{definition}[$\freq{b}{s}{\widehat{d_1}\ldots \widehat{d_{r-1}}}{\del}$] \label{def:freq-kquery}
    We define the shift-invariant frequency $\freq{b}{s}{\widehat{d_1}\ldots \widehat{d_{r-1}}}{\del}$ of a pattern $b\in \Sigma^r$ in a string $s \in \Sigma^\nu$ on the query differences $\widehat{d_1}\ldots \widehat{d_{r-1}}$ with deletion pattern $\del$ on strings of length $\nu$ as follows. 
    
    For a shift $i \in 1\ldots \nu$, let $\cD_i:=\cD[i, i+\widehat{d_1}\ldots i+\widehat{d_1}\ldots \widehat{d_{r-1}};\del]$ Then,
    \[
        \freq{b}{s}{\widehat{d_1}\ldots \widehat{d_{r-1}}}{\del} := \mathop{\bbE}_{i\in [\nu]} \Big[ ~\bbP \big[ s(\cD_i)=b \big] ~\Big].
    \]
\end{definition}

Next, we define the error probability of querying a list with differences $\widehat{d_1}\ldots \widehat{d_{r-1}}$ in a string $s$. Specifically, there's some chance of outputting $\err$ because some of the indices in the queries are out of bounds when the shift is realized. 

\begin{definition}[$\err(s;\widehat{d_1}\ldots \widehat{d_{r-1}};\del)$] \label{def:err-kquery}
    The error probability $\err(s;\widehat{d_1}\ldots \widehat{d_{r-1}};\del)$ when querying $i, i+\widehat{d_1}\ldots i+\widehat{d_1}+\ldots +\widehat{d_{r-1}}$ for a uniformly random shift $i \in [\nu]$ on a string $s$ with deletion pattern $\del$ is the probability that at least one index is out of bounds (and so the query returns $\err$).
\end{definition}
% \begin{definition}[$\Weight{b}{z}{\widehat{q_1}\ldots \widehat{q_r}}{\del}$] \label{def:W-kquery}    
%     We define the shift-invariant weight $\Weight{b}{z}{\widehat{q_1}\ldots \widehat{q_r}; \del}$ of a pattern $b\in \Sigma^r$ with deletion pattern $\del$ on strings of length $\nu$ in a string $z \in \Sigma^*$ whose length is denoted $\alpha\nu$ on the queries $\widehat{q_1}\ldots \widehat{q_r}$ as follows. 
%     \[
%         \Weight{b}{z}{\widehat{q_1}\ldots \widehat{q_r}}{\del} := \mathop{\bbE}_{i\in [\alpha]} \Big[ \freq{b}{C_{i,\nu}}{\widehat{q_1}\ldots \widehat{q_r}}{\del}^2 \Big].
%     \]

%     \noindent Moreover, we define the cumulative weight of the code $\C : \Sigma^n \to \Sigma^M$ as follows:
%     \[
%         \CWeight{b}{\C}{\widehat{q_1}\ldots \widehat{q_r}}{\del} := \mathop{\bbE}_{x\in \Sigma^n} \left[ \Weight{b}{\C(x)}{\widehat{q_1}\ldots \widehat{q_r}}{\del} \right].
%     \]
% \end{definition}

\subsubsection{Approximation Equivalence of Queries}

Next, we'll define what it means for a list of query differences $\widehat{\sd_1}\ldots \widehat{\sd_{r-1}}$ to approximate another list of differences $\widehat{d_1}\ldots \widehat{d_{r-1}}$ for a deletion pattern $\del$.

\begin{definition}[$\delta$-Approximate Queries]
    For a deletion pattern $\del$ on a string $s$ of length $\nu$, a list of query differences $\widehat{\sd_1}\ldots \widehat{\sd_{r-1}}$ $\delta$-approximates another list of query differences $\widehat{d_1}\ldots \widehat{d_{r-1}}$ if $\widehat{\sd_1}+\ldots+\widehat{\sd_{r-1}} \leq \widehat{d_1}+\ldots+\widehat{d_{r-1}}$ and for all $b\in \Sigma^r$,
    \begin{align*}
        \begin{aligned}
        \Big| \freq{b}{s}{\widehat{\sd_1}\ldots \widehat{\sd_{r-1}}}{\del} - \freq{b}{s}{\widehat{d_1}\ldots \widehat{d_{r-1}}}{\del} \Big|& \\
        + \err(s;\widehat{\sd_1}\ldots \widehat{\sd_{r-1}};\del) + \err(s;\widehat{d_1}\ldots \widehat{d_{r-1}};\del)&
        \end{aligned} ~< \delta.
    \end{align*}

    If $s$ is length $\alpha\nu$ and $\del$ is a deletion pattern on strings of length $\nu$, then we say that $\widehat{\sd_1}\ldots \widehat{\sd_{r-1}}$ $\delta$-approximates $\widehat{d_1}\ldots \widehat{d_{r-1}}$ if
    \begin{align*}
        \mathop{\bbE}_{i \in [\alpha]} \left[ 
        \begin{aligned}
            \Big| \freq{b}{C_{i,\nu}(s)}{\widehat{\sd_1}\ldots \widehat{\sd_{r-1}}}{\del} - \freq{b}{C_{i,\nu}(s)}{\widehat{d_1}\ldots \widehat{d_{r-1}}}{\del} \Big|& \\
            + \err(C_{i,\nu}(s);\widehat{\sd_1}\ldots \widehat{\sd_{r-1}};\del) + \err(C_{i,\nu}(s);\widehat{d_1}\ldots \widehat{d_{r-1}};\del)&
        \end{aligned}
        \right] ~< \delta.
    \end{align*}
    In other words, on the average chunk of length $\nu$, $\widehat{\sd_1}\ldots \widehat{\sd_{r-1}}$ $\delta$-approximates $\widehat{d_1}\ldots \widehat{d_{r-1}}$.
\end{definition}

% \begin{definition}[Approximate Equivalence of Queries for a Code]
%     For a family of deletion pattern $\del$ on strings of different lengths and if the code $\C$ has length $\alpha\nu$, a list of queries $\widehat{\sq_1}\ldots \widehat{\sq_r}$ $\delta$-approximates another list of queries $\widehat{q_1}\ldots \widehat{q_r}$ if the following holds. For each deletion pattern, for a randomly chosen chunk of length $\nu$ of a randomly chosen codeword with probability $1-\frac\delta2$, it holds that $\widehat{\sq_1}\ldots \widehat{\sq_r}$ $\frac\delta2$-approximates $\widehat{q_1}\ldots \widehat{q_r}$ on that chunk. More precisely, for a $\del$ in the family of deletion patterns, let $\nu$ denote the length of the string. It holds that
%     \[
%         \mathop{\bbP}_{x \in \Sigma^n, j} \Big[ \Big] > 1-\frac\delta2
%     \]
% \end{definition}
% \mnote{fix stuff about multiples}

\subsubsection{Layered Deletion Patterns and Significant Layers}

Next, we'll define a specific type of deletion pattern that we call a \emph{layered} deletion pattern.

\begin{definition}[Layered Deletion Pattern]
    For a length $\nu = \kappa^a$, we define a layered deletion pattern on the indices $[\nu]$. Each layer $\alpha \in [a]$ is assigned a maximum deletion fraction described by a function $f_\del: \bbN \to \bbR$. 
    
    To perform the (randomized) deletions, for each layer $\alpha \in [a]$, the string is divided into chunks of length $\kappa^\alpha$. Then, for each chunk, a uniformly random integer $I$ between $0$ and $\lfloor \eps_\alpha\kappa^\alpha \rfloor$ is chosen, and the first $I$ bits of the chunk are marked to be deleted. At the end, traverse the string left to write. Keep a counter that starts at $0$, and at every index, if an index is marked as deleted $t$ times, add $t$ to the counter. Moreover, if the counter is greater than $0$, delete the index and decrement the counter by $1$.

    The layered deletion pattern $\del_a$ is the pattern $\del$ restricted to layers $1\ldots a$. It is on strings of length $\kappa^a$ instead of $\kappa^M$. The layered deletion pattern $\del_{a|A}$ is the pattern $\del_a$ extended to strings of length $\kappa^A$ but still only incorporating the deletions in layers $1\ldots a$.
\end{definition}

Next, we'll define significant layers in a code $\C$. Loosely speaking, these are the values of $a$ such that performing some fixed list of queries with differences $\widehat{d_1}\ldots \widehat{d_{r-1}}$ with a random shift of about $\kappa^a$ generates different outputs from performing the queries with a random shift somewhat smaller than $\kappa^a$. These are the scales at which we'll need to insert random deletions.

\begin{definition} [Significant Layers for Threshold $\rho$]
    We define as follows a set $\cL$ to be the significant layers for some threshold $\rho$ given the code $\C$, a pattern $b \in \Sigma^r$, query differences $\widehat{d_1}\ldots \widehat{d_{r-1}}$ and deletion pattern $\del$.

    The possible layers are $1\ldots m$, corresponding to partitions of a codeword into exponentially-sized intervals. The first $10\kappa^\kappa$ significant layers will be the one corresponding to layer $\tau=\lfloor (\widehat{d_1}+ \ldots +\widehat{d_{r-1}}) \rfloor, \ldots , \lfloor (\widehat{d_1}+ \ldots +\widehat{d_{r-1}}) \rfloor + 10\kappa^\kappa-1$. Every future layer $\ell_{i+1}$ will be constructed from $\ell_i$ as follows. It is the smallest possible $\ell_{i+1} \in [m]$ satisfying 
    %both $\ell_{i+1}>\ell_i-2\log\rho^{-2}$ and also 
    \[
        \mathop{\bbE}_{z\in \text{Im}(\C), j \in \left[\frac{M}{\kappa^{\ell_i}}\right]} \Big[ ~\Big| \freq{b}{C_{j',\kappa^{\ell_{i+1}}}(z)}{\widehat{d_1}\ldots \widehat{d_{r-1}}}{\del_{\tau | \ell_{i+1}}} - \freq{b}{C_{j,\kappa^{\ell_i}}(z)}{\widehat{d_1}\ldots \widehat{d_{r-1}}}{\del_{\tau | \ell_i}} \Big|~ \Big] > \rho
    \]
    where $j' = \lfloor \frac{i}{\kappa^{\ell_{i+1}-\ell_i}} \rfloor$, so that $C_{j',\kappa^{\ell_i}}(z)$ is the length $\kappa^{\ell_i}$ chunk that contains $C_{j,\kappa^{\ell_{i+1}}}(z)$.
    
    Loosely speaking, $\ell_{i+1}$ is the smallest layer significantly larger than $\ell_i$ such that queries performed at the scale $\ell_{i+1}$ look significantly different to those performed at the scale $\ell_i$.
\end{definition}

\subsubsection{Useful Lemmas}

\begin{lemma} \label{lem:transitivity}
    Given a string $s\in \Sigma^*$ and a deletion pattern $\del$ on strings of length $\nu|s$, if a list of query differences $\widehat{d_1}\ldots \widehat{d_{r-1}}$ $\delta_1$-approximates another such list $\widehat{d'_1}\ldots \widehat{d'_{r-1}}$ and $\widehat{d'_1}\ldots \widehat{d'_{r-1}}$ $\delta_2$-approximates another list $\widehat{d''_1}\ldots \widehat{d''_{r-1}}$, then $\widehat{d_1}\ldots \widehat{d_{r-1}}$ $\delta_1+\delta_2$-approximates $\widehat{d''_1}\ldots \widehat{d''_{r-1}}$.
\end{lemma}

\begin{proof}
    This follows from the triangle inequality. Note that the condition is met that $\widehat{d_1}+\ldots +\widehat{d_{r-1}}< \widehat{d''_1}+\ldots +\widehat{d''_{r-1}}$. We have that 
     \begin{align*}
        \mathop{\bbE}_{i \in [\alpha]} \left[ 
        \begin{aligned}
            \Big| \freq{b}{C_{i,\nu}(s)}{\widehat{d_1}\ldots \widehat{d_{r-1}}}{\del} - \freq{b}{C_{i,\nu}(s)}{\widehat{d'_1}\ldots \widehat{d'_{r-1}}}{\del} \Big|& \\
            + \err(C_{i,\nu}(s);\widehat{d_1}\ldots \widehat{d_{r-1}};\del) + \err(C_{i,\nu}(s);\widehat{d'_1}\ldots \widehat{d'_{r-1}};\del)&
        \end{aligned}
        \right] ~< \delta_1
    \end{align*}
    and also that
    \begin{align*}
        \mathop{\bbE}_{i \in [\alpha]} \left[ 
        \begin{aligned}
            \Big| \freq{b}{C_{i,\nu}(s)}{\widehat{d'_1}\ldots \widehat{d'_{r-1}}}{\del} - \freq{b}{C_{i,\nu}(s)}{\widehat{d
            ''_1}\ldots \widehat{d''_{r-1}}}{\del} \Big|& \\
            + \err(C_{i,\nu}(s);\widehat{d'_1}\ldots \widehat{d'_{r-1}};\del) + \err(C_{i,\nu}(s);\widehat{d''_1}\ldots \widehat{d''_{r-1}};\del)&
        \end{aligned}
        \right] ~< \delta_2.
    \end{align*}
    Therefore, by the triangle inequality and the fact that the $\err$ function is always nonnegative, 
    \begin{align*}
        \mathop{\bbE}_{i \in [\alpha]} \left[ 
        \begin{aligned}
            \Big| \freq{b}{C_{i,\nu}(s)}{\widehat{d_1}\ldots \widehat{d_{r-1}}}{\del} - \freq{b}{C_{i,\nu}(s)}{\widehat{d
            ''_1}\ldots \widehat{d''_{r-1}}}{\del} \Big|& \\
            + \err(C_{i,\nu}(s);\widehat{d_1}\ldots \widehat{d_{r-1}};\del) + \err(C_{i,\nu}(s);\widehat{d''_1}\ldots \widehat{d''_{r-1}};\del)&
        \end{aligned}
        \right] ~< \delta_1+\delta_2.
    \end{align*}
    
\end{proof}

% \begin{lemma} \label{lem:errs}
%     Given any deletion pattern $\del$ that deletes at most $2^{-\kappa}$-fraction of and a list of query differences $\widehat{d_1}\ldots \widehat{d_{r-1}}$ to some string $z$, the following holds: 
%     \[
%         \err(z;\widehat{d_1}\ldots \widehat{d_{r-1}};\del) \leq 2^{-\kappa} + \frac{\widehat{d_1}+\ldots +\widehat{d_{r-1}}}{|z|}
%     \]
% \end{lemma}

% \begin{proof}
    
% \end{proof}

\begin{lemma} \label{lem:const-sig-layers}
    Consider a layered deletion pattern $\del$ on the code $\C$ whose total corruption is at most $2^{-\kappa}$. For a list of query differences $\widehat{d_1}+ \ldots +\widehat{d_{r-1}}$ and a threshold $\rho>2^{-\kappa}$, there are at most $\kappa^\kappa+\rho^{-3}$ significant layers in $\C$.
\end{lemma}

\begin{proof}
    Define $a:=\lfloor (\widehat{d_1}+ \ldots +\widehat{d_{r-1}}) \rfloor$. The first $\kappa^\kappa$ subsequent layers after $a$ are all significant layers. 

    For any $b\in \Sigma^r$, $z\in \text{Im}(\C)$, and layer $\ell$, we define the weight as follows.
    \[
        \Weight{b}{z}{\widehat{d_1}\ldots \widehat{d_{r-1}}}{\del}{\ell} := \mathop{\bbE}_{i\in [\ell]} \Big[ \freq{b}{C_{i,\kappa^\ell}(z)}{\widehat{d_1}\ldots \widehat{d_{r-1}}}{\del_{a|\ell}}^2 \Big].
    \]
    Moreover, we define the cumulative weight of the code $\C$ as follows:
    \[
        \CWeight{b}{\C}{\widehat{d_1}\ldots \widehat{d_{r-1}}}{\del}{\ell} := \mathop{\bbE}_{x\in \Sigma^n} \Big[ \Weight{b}{\C(x)}{\widehat{d_1}\ldots \widehat{d_{r-1}}}{\del}{\ell} \Big].
    \]

    Assume for the sake of contradiction that there are more than $c=\rho^{-5}$ remaining significant layers aside from the first $\kappa^\kappa$, denoted $\ell_1\ldots \ell_c$. Recall that a significant layer satisfies: 
    \[
        \mathop{\bbE}_{z\in \text{Im}(\C), j \in \left[\frac{M}{\kappa^{\ell_i}}\right]} \Big[ ~\Big| \freq{b}{C_{j',\kappa^{\ell_{i+1}}}(z)}{\widehat{d_1}\ldots \widehat{d_{r-1}}}{\del_{\tau | \ell_{i+1}}} - \freq{b}{C_{j,\kappa^{\ell_i}}(z)}{\widehat{d_1}\ldots \widehat{d_{r-1}}}{\del_{\tau | \ell_i}} \Big|~ \Big] > \rho
    \]
    where $i' = \lfloor \frac{i}{\kappa^{\ell_{i+1}-\ell_i}} \rfloor$, so that $C_{j',\kappa^{\ell_i}}(z)$ is the length $\kappa^{\ell_i}$ chunk that contains $C_{j,\kappa^{\ell_{i+1}}}(z)$.

    This implies 
    \[
        \mathop{\bbE}_{z\in \text{Im}(\C), j \in \left[\frac{M}{\kappa^{\ell_i}}\right]} \Bigg[ ~\Big( \freq{b}{C_{j',\kappa^{\ell_{i+1}}}(z)}{\widehat{d_1}\ldots \widehat{d_{r-1}}}{\del_{\tau | \ell_{i+1}}} - \freq{b}{C_{j,\kappa^{\ell_i}}(z)}{\widehat{d_1}\ldots \widehat{d_{r-1}}}{\del_{\tau | \ell_i}} \Big)^2~ \Bigg] > \rho^2
    \]
    by the QM-AM inequality. Expanding, we get the the left side is equal to
    \begin{align*}
        &\CWeight{b}{\C}{\widehat{d_1}\ldots \widehat{d_{r-1}}}{\del}{\ell_i} - \CWeight{b}{\C}{\widehat{d_1}\ldots \widehat{d_{r-1}}}{\del}{\ell_{i+1}} \\ 
        +
        &\mathop{2\bbE}_{z\in \text{Im}(\C), j \in \left[\frac{M}{\kappa^{\ell_i}}\right]} \left[~ 
        \begin{aligned}
            &\freq{b}{C_{j',\kappa^{\ell_{i+1}}}(z)}{\widehat{d_1}\ldots \widehat{d_{r-1}}}{\del_{\tau | \ell_{i+1}}} \cdot \\
            &
            \Big(
            \freq{b}{C_{j',\kappa^{\ell_{i+1}}}(z)}{\widehat{d_1}\ldots \widehat{d_{r-1}}}{\del_{\tau | \ell_{i+1}}}
            -
            \freq{b}{C_{j,\kappa^{\ell_i}}(z)}{\widehat{d_1}\ldots \widehat{d_{r-1}}}{\del_{\tau | \ell_i}} \Big) 
        \end{aligned}
        ~\right]
    \end{align*}
    Finally, note that within a block of length $\kappa^{\ell_{i+1}}$, the values of $\freq{b}{C_{j',\kappa^{\ell_{i}}}(z)}{\widehat{d_1}\ldots \widehat{d_{r-1}}}{\del_{\tau | \ell_{i}}}$, which just represent the probability of getting the output $b$ on a randomly shifted query with the specified differences, average to $\freq{b}{C_{j',\kappa^{\ell_{i+1}}}(z)}{\widehat{d_1}\ldots \widehat{d_{r-1}}}{\del_{\tau | \ell_{i+1}}}$ minus an error term corresponding to the chance that the output was $\err$ in the case of size $\kappa^{\ell_i}$ blocks due to out of bounds but is $b$ when the block size is expanded, which we'll denote $\extra{b}{C_{j',\kappa^{\ell_{i+1}}}(z)}{\widehat{d_1}\ldots \widehat{d_{r-1}}}{\del_{\tau | \ell_{i+1}}}$. This term is nonnegative and at most the probability of an out of bounds error, which is at most $\kappa^{-\kappa^\kappa}$ since there are $-\kappa^\kappa$ layers before $\ell_i$.

    Thus, the above equation is equal to
    \begin{align*}
        &\CWeight{b}{\C}{\widehat{d_1}\ldots \widehat{d_{r-1}}}{\del}{\ell_i} - \CWeight{b}{\C}{\widehat{d_1}\ldots \widehat{d_{r-1}}}{\del}{\ell_{i+1}} \\ 
        &+
        \mathop{2\bbE}_{z\in \text{Im}(\C), j' \in \left[\frac{M}{\kappa^{\ell_{i+1}}}\right]} \left[~ 
        \begin{aligned}
            &\freq{b}{C_{j',\kappa^{\ell_{i+1}}}(z)}{\widehat{d_1}\ldots \widehat{d_{r-1}}}{\del_{\tau | \ell_{i+1}}}~
            \cdot \\
            &\extra{b}{C_{j',\kappa^{\ell_{i+1}}}(z)}{\widehat{d_1}\ldots \widehat{d_{r-1}}}{\del_{\tau | \ell_{i+1}}}
        \end{aligned}
        ~\right] \\
        <~
        &\CWeight{b}{\C}{\widehat{d_1}\ldots \widehat{d_{r-1}}}{\del}{\ell_i} - \CWeight{b}{\C}{\widehat{d_1}\ldots \widehat{d_{r-1}}}{\del}{\ell_{i+1}} \\ 
        &+
        \mathop{2\bbE}_{z\in \text{Im}(\C), j' \in \left[\frac{M}{\kappa^{\ell_{i+1}}}\right]} \left[~ 
            \extra{b}{C_{j',\kappa^{\ell_{i+1}}}(z)}{\widehat{d_1}\ldots \widehat{d_{r-1}}}{\del_{\tau | \ell_{i+1}}}
        ~\right] \\
        <~
        &\CWeight{b}{\C}{\widehat{d_1}\ldots \widehat{d_{r-1}}}{\del}{\ell_i} - \CWeight{b}{\C}{\widehat{d_1}\ldots \widehat{d_{r-1}}}{\del}{\ell_{i+1}} + \kappa^{-\kappa^\kappa}
    \end{align*}

    Therefore, 
    \[
        \CWeight{b}{\C}{\widehat{d_1}\ldots \widehat{d_{r-1}}}{\del}{\ell_i} - \CWeight{b}{\C}{\widehat{d_1}\ldots \widehat{d_{r-1}}}{\del}{\ell_{i+1}}  > \rho^2 - \kappa^{-\kappa^\kappa} > \rho^3.
    \]
    If there are more than $\rho^{-3}$ terms in the sequence $\ell_i$, it would hold that $\CWeight{b}{\C}{\widehat{d_1}\ldots \widehat{d_{r-1}}}{\del}{\ell_1}-\CWeight{b}{\C}{\widehat{d_1}\ldots \widehat{d_{r-1}}}{\del}{\ell_c}>1$, but all values of the weight function fall between $0$ and $1$, which is a contradiction.
\end{proof}

% \begin{definition}[$\Weight{b}{z}{\widehat{q_1}\ldots \widehat{q_r}}{\del}$] \label{def:W-kquery}    
%     We define the shift-invariant weight $\Weight{b}{z}{\widehat{q_1}\ldots \widehat{q_r}; \del}$ of a pattern $b\in \Sigma^r$ with deletion pattern $\del$ on strings of length $\nu$ in a string $z \in \Sigma^*$ whose length is denoted $\alpha\nu$ on the queries $\widehat{q_1}\ldots \widehat{q_r}$ as follows. 
%     \[
%         \Weight{b}{z}{\widehat{q_1}\ldots \widehat{q_r}}{\del} := \mathop{\bbE}_{i\in [\alpha]} \Big[ \freq{b}{C_{i,\nu}}{\widehat{q_1}\ldots \widehat{q_r}}{\del}^2 \Big].
%     \]

%     \noindent Moreover, we define the cumulative weight of the code $\C : \Sigma^n \to \Sigma^M$ as follows:
%     \[
%         \CWeight{b}{\C}{\widehat{q_1}\ldots \widehat{q_r}}{\del} := \mathop{\bbE}_{x\in \Sigma^n} \left[ \Weight{b}{\C(x)}{\widehat{q_1}\ldots \widehat{q_r}}{\del} \right].
%     \]
% \end{definition}

% \begin{lemma} \label{lem:approx-higher-layers}
%     Given a string $z \in \Sigma^M$ and a layered deletion pattern $\del$ with total corruption at most $2^{-\kappa}$, for some $a$, let us assume $\widehat{\sd_1}\ldots \widehat{\sd_{r-1}}$ $\delta$-approximates $\widehat{d_1}\ldots \widehat{d_{r-1}}$ for the deletion pattern $\del_a$. Then, for any $a'>a$, the same holds with a $(\delta+2^{-\kappa})$-approximation for the deletion pattern $\del_{a'}$.
% \end{lemma}

% \begin{proof}
    
% \end{proof}

\subsubsection{Lemmas for Proof of Lemma~\ref{lem:main}}

The next three lemmas are the core of the proof of Lemma~\ref{lem:main} in Section~\ref{sec:main-lemma}. They describe modifications that can be made to a list of queries $\widehat{d_1}\ldots \widehat{d_{r-1}}$ to make a new list of queries $\widehat{\sd_1}\ldots \widehat{\sd_{r-1}}$ that approximates the old one.

The first lemma essentially says that shifting a query difference by a small amount nearby a layer $a$ where $f_\del(a)>0$ does not affect the distribution of outputs.

\begin{lemma} \label{lem:corrupted-layers}
    Consider a string $z \in \Sigma^M$, a layered deletion pattern $\del$ on strings of length $\kappa^a$ for some $a < m$ with total corruption fraction at most $2^{-\kappa}$, and query differences $\widehat{d_1}\ldots \widehat{d_{r-1}}$ satisfying $\widehat{d_1}+\ldots +\widehat{d_{r-1}}<\kappa^{a-\kappa^2}$. Fix $1\leq j < r$ and let $a' <  \log_\kappa (\widehat{d_j}) -\kappa^{100}$ be a layer such that $f_\del(a') \geq \kappa^{-\kappa^3r^r}$. Let $0\leq \tau < f_\del(a')\cdot \kappa^{a'-\kappa^{100}}$. Then, $\widehat{d_1}\ldots \widehat{d_j}-\tau \ldots \widehat{d_{r-1}}$ $(2^{-\kappa/2})$-approximates $\widehat{d_1}\ldots \widehat{d_{r-1}}$ on $\del$.
\end{lemma}

\begin{proof}
    Let $z_a$ be a chunk of length $\kappa^a$ of the string $z$. Let $\widehat{\sd_1}\ldots \widehat{\sd_{r-1}}:=\widehat{d_1}\ldots \widehat{d_j}-\tau \ldots \widehat{d_{r-1}}$. Let $q(\widehat{\sd_1}\ldots \widehat{\sd_{r-1}}) \in [\kappa^a]^r$ be the distribution of $\sq_1\ldots \sq_r$ in the original string when the queries $\widehat{\sd_1}\ldots \widehat{\sd_{r-1}}$ with a random shift are performed to the string $z_a$ with random deletions according to $\del$. Similarly, define $q(\widehat{d_1}\ldots \widehat{d_{r-1}})$. It suffices to show that the TV distance of these distributions is at most $(2^{-\kappa/2})$ (where getting out of bounds errors counts against the TV distribution on both ends).

    Consider a randomly generated set of deletions from $\del$. This assigns for each $\ell\leq a$, a list of $\kappa^{a-\ell}$ values $\cE_{\ell,i}$ for how many indices to delete in each layer using which the layered deletion algorithm is performed. Moreover, condition on the randomness of the queries' shift. Denote the induced indices in the original codeword when the queries $\widehat{d_1}\ldots \widehat{d_{r-1}}$ with the fixed random shift are performed by $q_1\ldots q_r$. The way we found these was by considering the string $z$, marking all the deleted symbols, and counting $\widehat{q_1}$ non-deleted symbols from the left, then $\widehat{d_1}$ more symbols, and so on.

    Let $I$ be the index of the first chunk of length $\kappa^{a'}$ that begins after $q_j$. Increment the value of $\cE_{a',I}$ by $\tau$. The probability that this becomes invalid (note that all valid options are equally likely since $\del$ is uniform) option for $\del$ is at most 
    \[
        \frac{\tau}{f_\del(a')\kappa^{a'}} \leq \frac{f_\del(a')\cdot \kappa^{a'-\kappa^{100}}}{f_\del(a')\kappa^{a'}} \leq \kappa^{\kappa^{100}}.
    \]
    The above defines a mapping from an element of $\del$ to another element of $\del$, except for a very small fraction of the time. In this new element of $\del$, the queries $\widehat{\sq_1}\ldots \widehat{\sq_r}$ yield $q_1\ldots q_r$, because the transformation just causes an extra $\tau$ indices to get marked between $q_j$ and $q_{j+1}$ which exactly compensates for subtracting $\tau$. To show that the new likelihood of getting $q_1\ldots q_r$ is at least as large as before, we need to show this mapping is an injection. In particular, no instantiation of $\del$ is mapped to twice.
    
    We will show the mapping is injective. If two instances of $\del$ and the query shift map to the same thing, they can only differ on the layer $a'$ because only that changes in the mapping. If the index $I$ is the same, then the map of adding $\tau$ is injective. If the index $\tau$ is different, then something before reaching $I$ must have been different, to induce a different choice of $I$ (which depends only on the choices of $\cE$ before $I$), but then that index of $\cE$ will be different in the image as well.

    The two failure cases are where the mapping fails, because $\cE_{\ell,i}+\tau$ is too large, or the initial choice of $\del$ and the query shift resulted in $\err$ rather than a valid output. The latter occurs with probability at most $\frac{\widehat{d_1}+\ldots+\widehat{d_{r-1}}}{\kappa^a} < \kappa^{-\kappa^2}$. This failure probability amounts to less than $\kappa^{-\kappa}$ in total.

    Therefore, the TV distribution of the queries differs by at most $2\kappa^{-\kappa}$, since the one-sided TV distance differs by at most $\kappa^{-\kappa}$.
\end{proof}

The second lemma essentially says that replacing a group of query differences with one that approximates it, when the surrounding query differences are much larger in comparison, does not affect the distribution of outputs.

\begin{lemma} \label{lem:replace-middle}
    Consider a string $z \in \Sigma^M$, a layered deletion pattern $\del$ on strings of length $\kappa^a$ for some $a < m$ with total corruption fraction at most $2^{-\kappa}$, and query differences $\widehat{d_1}\ldots \widehat{d_{r-1}}$. Moreover, every layer has corruption at most $\kappa^{-\kappa^2}$ and $\widehat{d_1}+\ldots +\widehat{d_{r-1}}<\kappa^{a-\kappa^2}$. 
    
    Suppose there exists $1\leq j \leq j' \leq r-1$ and $a' \in \bbN$ such that $f_\del(a'+\kappa^{\kappa/1000})>\kappa^{-\kappa^{\kappa/2000}}$
    \[
    (\widehat{d_j} +\ldots + \widehat{d_{j'}}) \cdot \kappa^{\kappa^{\kappa/100}} < \kappa^{a'} < \min (\widehat{d_{j-1}}, \widehat{d_{j'+1}})\cdot \kappa^{-\kappa^{\kappa/100}}.
    \]
    Select $\widehat{\sd_j}\ldots \widehat{\sd_{j'}}$ such that it $((2k)^{j'-j+1}2^{-\kappa/100})$-approximates $\widehat{d_j} \ldots \widehat{d_{j'}}$ on the deletion pattern $\del_{a'}$. Then $\widehat{d_1}\ldots \widehat{\sd_j}\ldots \widehat{\sd_{j'}} \ldots \widehat{d_{r-1}}$ $((2k)^{r-1}2^{-\kappa/100})$-approximates $\widehat{d_1}\ldots \widehat{d_{j'-j+1}}$ on $\del$.
\end{lemma}

\begin{proof}
    We use a similar proof strategy as Lemma~\ref{lem:corrupted-layers}. Our goal will be to find an injective mapping from the queries $\widehat{d_1}\ldots \widehat{d_{r-1}}$ to $\widehat{\sd_1}\ldots \widehat{\sd_{r-1}}$ resulting in the same distribution of outputs that only fails with small probability over the randomness of choosing a query shift and an instance of $\del$.

    First, we remark that there is an injective map that except with failure probability $((2k)^{j'-j+1}2^{-\kappa/100})$ over uniformly randomly chosen options does the following. It maps a size $\kappa^{a'}$ chunk (of the $\kappa^{a-a'}$ options), query shift for that chunk, and $\del_{a'}$ pattern to another such that the output of $\widehat{d_j}\ldots \widehat{d_{j'}}$ is the same as $\widehat{\sd_j}\ldots \widehat{\sd_{j'}}$.

    Secondly, we consider the distribution of a random instance of $\del_a$ and query shift. We look at where $q_j$ falls, and our main goal will be to show that the chunk of size $\kappa^{a'}$ it is in, as well as the deletion pattern within that chunk, is essentially uniform (small TV distance from uniform). We'll start by conditioning on the deletions for chunks of size larger than $\kappa^{a'}$. The corruptions on chunks larger than $\kappa^{a'+2\kappa}$, when implemented, only affect up to $\kappa^{-\kappa}$-fraction of the chunks at all, and so ruling those out is only a negligible fraction of total chunks. The ones between $\kappa^{a'}$ and $\kappa^{a'+2\kappa}$ only cause at most $\kappa^{a'+2\kappa-\kappa^2}<\kappa^{a'-\kappa}$ deletions total, and so only affect a negligible $\kappa^{-\kappa}$ fraction of the chunk they fall into. Now, generate the small chunks' deletions randomly. The query shift dominates in which chunk $q_j$ will fall into of the ones that weren't ruled out by the large deletions or by being in the first $\widehat{d_1}+\ldots+\widehat{d_j}$ eligible bits. In particular, the size remaining post-deletion of a chunk determines the chance of falling into it given the query shift, which differs by a factor of at most $1-2^{-\kappa}-\kappa^{-\kappa}$ between chunks. Thus, the distribution is TV distance at most $2^{-\kappa/2}$ from uniform.

    Combining these two facts, we can apply the injective map to adjust the deletion pattern in the chunk that $\widehat{d_j}$ falls into so that queries $\widehat{d_j}\ldots \widehat{d_{j'}}$ become $\widehat{\sd_j}\ldots \widehat{\sd_{j'}}$ except with probability $((2k)^{j'-j+1}2^{-\kappa/100}) + 2^{-\kappa/2}$. In order to account for the query shift in the chunk of length $\kappa^{a'}$, one has to also adjust the amount of corruption in the length $\kappa^{a'+\kappa^{\kappa/1000}}$ length chunks directly before and after the chunk where $\widehat{d_j}$ fell. This adjustment fails with probability at most 
    \[
        \frac{\kappa^{a'}}{\kappa^{a'+\kappa^{\kappa/1000}}f_\del(a'+\kappa^{\kappa/1000})} < \kappa^{-\kappa/2}.
    \]

    Because this map is reversible, it is injective. In all, the failure probability is at most $((2k)^{j'-j+1}2^{-\kappa/100}) + 2^{-\kappa/10}$, and so the TV distance of distribution of outputs for the two sets of queries is at most two times that, which in total results in a $((2k)^{j'-j+1}2^{-\kappa/100}) + 2^{-\kappa/10} < (2k)^{r-1}$ approximation.
\end{proof}

\begin{lemma} \label{lem:no-sig-layers}
    Consider a layered deletion pattern $\del$ on strings of length $\kappa^a$ for some $a< m$ with total corruption fraction at most $2^{-\kappa}$, and query differences $\widehat{d_1}\ldots \widehat{d_{r-1}}$. Choose $a'$ such that every $\widehat{d_i}$ is either larger than $\kappa^{a'}\cdot \kappa^{\kappa^{\kappa/100}}$ or smaller than $\kappa^{a'}\cdot \kappa^{-\kappa^{\kappa/100}}$. Let us further assume that for all lists of query differences of length $s$ at most $r-2$ (corresponding to $r-1$ queries) where the total width is at most $\kappa^{a'}\kappa^{-\kappa^2}$, there exists a list of query differences that $(2k)^{s}(2^{-\kappa/2})$-approximates it on at least $(1-(2k)^{s}2^{\kappa/2})$ fraction of strings $z\in \text{Im}(\C)$ for the deletion pattern $\del_{a'}$. Finally, assume that for at most 
    \[
        \mathop{\bbE}_{j \in \left[\frac{M}{\kappa^{a'-\kappa^2}}\right]} \left[ ~
        \left| 
        \begin{aligned}
            \freq{b}{C_{j',\kappa^{\widehat{d_1}+\ldots+\widehat{d_{r-1}}+\kappa^2}}(z)}{\widehat{d_1}\ldots \widehat{d_{r-1}}}{\del_{\tau | \widehat{d_1}+\ldots+\widehat{d_{r-1}}+\kappa^2}} \\
            - \freq{b}{C_{j,\kappa^{a'-\kappa^2}}(z)}{\widehat{d_1}\ldots \widehat{d_{r-1}}}{\del_{\tau | a'-\kappa^2}}
        \end{aligned}
        \right|
        ~ \right] < 2^{-\kappa/2}
    \]
    where $j' = \left\lfloor \frac{i}{\kappa^{\widehat{d_1}+\ldots+\widehat{d_{r-1}}+\kappa^2-a'-\kappa^2}} \right\rfloor$.
    
    has no significant layers between $a'-\kappa^2$ and $\widehat{d_1}+\ldots+\widehat{d_{r-1}}+\kappa^2$. Finally, $f_\del(a')>\kappa^{\kappa^{\kappa/1000}}$.

    %For $i\in [r-1]$, define $\widehat{\sd_i}:= \widehat{d_i}$ if $\widehat{d_i}<\kappa^{a'}$ and otherwise define $\widehat{\sd_i}:= \kappa^{a'}\kappa^{\kappa}$. 
    Then, there is a list of queries $\widehat{\sd_1}\ldots \widehat{\sd_{r-1}}$ that $(2k)^{r-1}(2^{-\kappa/100})$-approximates $\widehat{d_1}\ldots \widehat{d_{r-1}}$ on $\del$, which satisfies that all $\widehat{\sd_i}\leq \kappa^{a'+\kappa^{\kappa/1000}}$.
\end{lemma}

\begin{proof}
    We can partition the queries $\widehat{d_1}\ldots \widehat{d_{r-1}}$ into groups $\widehat{d_j}\ldots \widehat{d_{j'}}$ that comprise of the small query differences (and exclude the large differences). Firstly, we can apply Lemma~\ref{lem:replace-middle} to successively replace the query differences $\widehat{d_j}\ldots \widehat{d_{j'}}$ with a list $\widehat{\sd_j}\ldots \widehat{\sd_{j'}}$ from the list of approximating queries that has no significant layers in the desired range, such that $(2k)^{j'-j+1}2^{-\kappa/100}$-approximates it. We can do this by choosing $a'$ in the lemma to be $a'+\kappa^{\kappa/1000}$, because the adjacent query differences are all larger by at least a factor of $\kappa^{\kappa^{\kappa/100}}$, and the layer $a'$ has sufficient corruption. Each iteration of this is a $(2k)^{j'-j+1}2^{-\kappa/100}$-approximation, so by Lemma~\ref{lem:transitivity}, after performing all the replacements, we are left with a $(2k)^{r-1}2^{-\kappa/100}$-approximation at worst.

    Now, the queries in the groups $\widehat{\sd_j}\ldots \widehat{\sd_{j'}}$ has essentially the same distribution of outputs in the layers $a'-\kappa^2$ and $\widehat{d_1}+\ldots+\widehat{d_{r-1}}+\kappa^2$, and in fact in all layers between. At a high level, we want to show that replacing the remaining differences (large differences) with $\kappa^{a'+\kappa}$ does not change the distribution of outputs by much because there are no significant layers for any of the lists of short queries between layers $a'-\kappa^2$ and $\widehat{d_1}+\ldots+\widehat{d_{r-1}}+\kappa^2$, and so all query differences at approximately that scale function similarly.

    Let the groups of short queries be $\widehat{d_{j_1}}\ldots \widehat{\sd_{j'_1}}, \widehat{\sd_{j_2}}\ldots \widehat{\sd_{j'_2}}, \ldots$, such that between each pair, there is a single $\widehat{d_{p_1}}$ that is large. (If there are multiple large differences in a row, the short query group between them is the empty list.)
    
    We will show that we can replace any $\widehat{d_{p_i}}$ with $\kappa^{a'+\kappa^{\kappa/1000}}$ and adjust $\widehat{d_{p_{i+1}}}$ by $\widehat{d_{p_i}}-\kappa^{a'+\kappa^{\kappa/1000}}$ without changing the output distribution. In the case where $\widehat{d_{p_i}}$ is the rightmost difference, no further value $\widehat{d_{p_{i+1}}}$ need be adjusted We remark that this does not result in the same distribution of output queries induced in the original codeword, just of outputs in $\Sigma^r$. Repeating this process for each $\widehat{d_{p_i}}$ from left to right recovers the lemma statement.

    By the same argument as in the proof of Lemma~\ref{lem:replace-middle}, the distribution of which $\kappa^{a'}$ block that $\widehat{\sd_{j_i}}\ldots \widehat{\sd_{j'_i}}$ falls into and the corruption pattern within the block and the shift is at most $2^{-\kappa/2}$ TV distance from uniform. Then, when we reduce the value of $\widehat{d_{p_i}}$, there is a change in which block this query falls in. However, the distribution of where the list of queries falls now is still approximately uniform, since the argument would've also held using $\kappa^{a'+\kappa^{\kappa/1000}}$ for $\widehat{d_{p_i}}$. Then, the distributions of outputs for the queries induced by the differences $\widehat{\sd_{j_i}}\ldots \widehat{\sd_{j'_i}}$ is the same in both cases. Moreover, there is a bijection between the shift and instantiation of $\del_{a'}$ in the block that the list of queries falls into in each case, because there are no significant layers. Since there are no significant layers, the output distribution in each length $\kappa^{a'}$ block is the same, except with small probability at most $2^{\kappa/2}$. By increasing the value of $\widehat{d_{p_{i+1}}}$ by exactly how much $\widehat{d_{p_i}}$ was decreased, the distribution of the other outputs does not change. However, the conditional distribution may change, based on the shift induced by $\widehat{\sd_{j'_i}}$, but because the outputs could be approximately bijected without adjusting the shift of $\widehat{\sd_{j'_i}}$ by more than $\kappa^{a'}$, this can be compensated for by adjusting the deletion in the preceding and following $\kappa^{a'+\kappa^{\kappa/2000}}$-sized block. Formally, like in the previous arguments, the map where these adjustments are made forms an injection that succeeds with high probability.

    In all, this final list of queries, formed by replacing all the large differences by the value $\kappa^{a'+\kappa^{\kappa/1000}}$ results in a $(2k)^{r-1}(2^{-\kappa/100})$-approximation.
\end{proof}

% \begin{lemma} \label{lem:}
%     Given a string $s \in \Sigma^*$ of length $\kappa^a$ and a layered deletion pattern $\del$ on strings of length $\kappa^a$ with total corruption at most $2^{\kappa}$, given a pair of queries $\widehat{q_1}$ and $\widehat{q_2}$, the probability that the following holds is at least $1-\frac{q_2-q_1}{\kappa^a}-2^{\kappa}$. The distribution of what 
% \end{lemma}

% \begin{proof}

% \end{proof}

\subsection{Deletion Pattern and a Representative Set of Queries} \label{sec:del-pattern}

In this section, we'll define the layered deletion pattern $\del$ used by the adversary in terms of the choice of code $\C$.

The key idea of the proof of the main theorem (that a constant number of queries don't suffice for recovering every index) is that, given the deletion pattern $\del$, there is a constant-sized representative set of lists of queries $\cF$ such that any arbitrary query list can be simulated by an element of $\cF$. That is, the output of the arbitrary list of queries and the chosen list from $\cF$ have approximately the same distribution for most codewords in the code upto a random shift of the queries. More granularly, we'll let $\cF_r$ be the subset of $\cF$ consisting of successive differences for size $r$ query lists. 

In the following definition, we'll build up $\del$ and $\cF$ simultaneously layer by layer.

\longdef{{The Deletion Pattern $\del$ and Representative Family of Query Differences $\cF$}}{del-cF}{

    \noindent We define $\del$ by assigning for all $a\in [m]$ a value to $f_\del(a)$. Simultaneously, we'll build $\cF$.

    \medskip

    \noindent \textbf{Preliminary values for layers $a\leq 2\kappa^{\kappa^2}$:}
    For $a\leq \kappa^3$ let $f_\del(a)=0$ and for $\kappa^3 < a < 2\kappa^\kappa$, let $f_\del(a) := \kappa^{-\kappa^2}$. Correspondingly, add the following elements to $\cF_r$. Define the set $S$ of important differences as follows. Add every integer from $0\ldots \kappa^3$. Next, for $\kappa^3 < a < 2\kappa^\kappa$ and $0<i\leq \kappa^{2\kappa^2}$, add $\kappa^a\cdot \left( 1+ \frac{i}{\kappa^{2\kappa^2}}\right)$. To each $\cF_r$, add every $(r-1)$-sized list where each $\widehat{d_i}$ is an element of $S$. Also, include the empty list in $\cF_1$.

    \medskip
    
    \noindent \textbf{Assigning layers $a\in [2\kappa^{\kappa^2},m-\kappa^\kappa]$ (roughly) in sequence:}
    We'll proceed one layer at a time starting from $a=\kappa^2$ to add more corruption to $\del$. Which layers are corrupted next will depend on the current $\cF$ and $\del$. For each value of $a$ ranging from $\kappa^2$ to $m$, do the following:

    \smallskip

    \begin{enumerate}
    
    \item For each list of queries in $\widehat{\sq_1}\ldots \widehat{\sq_r} \in \cF$ with $\widehat{\sd_1}+\ldots+ \widehat{\sd_{r-1}} < \kappa^{a-\kappa^\kappa/2}$ and each $b\in \Sigma^r$, decide whether $a$ is a significant layer for any $b\in \Sigma^r$ for the deletion process $\del$ (only looking at what it is so far) for the threshold $2^{-\kappa}$.
    
    \item If not, then move on to the next $a$. Otherwise, let $r$ be the smallest for which there is a query list higher than the threshold. Assign $f_\del(a'):=\kappa^{-\kappa^3 r^r}$ for all $a-\kappa^\kappa \leq a' < a+\kappa^\kappa$ (unless the already-assigned value was higher, in which case leave it).
    
    \item For all $r'>r$ (note the strict inequality), add the following query difference lists to $\cF_r$. Define the set $S_a$ of important differences as follows: for $a-\kappa^\kappa \leq a' < a+\kappa^\kappa$ and $0<i\leq \kappa^{2\kappa^2r^r}$, add in $\kappa^a\cdot \left( 1+ \frac{i}{\kappa^{2\kappa^2r^r}}\right)$. Then, choose an arbitrary difference list $\widehat{\sd_1}\ldots \widehat{\sd_t}\in \cF$ with $t<r-1$ (for example, the singleton query list $\widehat{q_1}=1$) and set $\widehat{d_1}\ldots \widehat{d_{t}}$ to those values. Then, set $\widehat{d_{t+1}}$ to an element in $S_a$, and set $\widehat{q_{t+2}}\ldots \widehat{q_{t'}}$ as an arbitrary element of $\cF$ as before shifted such that the first query is $\widehat{q_{t+1}}$). Then set the next value $\widehat{d_{t'+1}}$ to an element in $S_a$ and repeat.

    \end{enumerate}

    \medskip

    \noindent \textbf{Assigning layers $a>m-2\kappa^\kappa$:} Assign $f_\del(a):=\kappa^{-\kappa^3}$. (This part will be completed after the previous one, even though there is a slight overlap in layer.) Assign $f_\del(m):=\kappa^{-5}$.
    
}

Now, we'll prove a couple facts about $\del$ and $\cF$.

\begin{lemma}
    For each $r$, the size of $\cF_r$ is at most $\kappa^{4\kappa^2r^r}$. 
\end{lemma}

\begin{proof}
    We'll show the claim by induction on $r$. For $r=1$, the only query we have added is $q_1=1$ so there is only one element.
    
    Now we'll show the claim for $r$ if it's true for all $r'<r$. 
    
    The first elements added to $\cF_r$ were determined by a set $S$ defined as every integer from $0\ldots \kappa^3$ and for $\kappa^3 < a < 2\kappa^{\kappa}$ and $0<i\leq \kappa^{2\kappa^2}$, adding an element. The size of this set is less than $\kappa^{3\kappa^2}$. Then, each of the $r-1$ differences between adjacent queries may be chosen from this set, for a total of $\left( \kappa^{3\kappa^2} \right)^r$ possible lists of queries. Note for clarity that this part did not involve the inductive hypothesis.

    In the second way of adding elements to $\cF_r$, we first find a value of $a$ such that the layer is significant for one of the queries in $\cF_{r'}$ with $r'<r$. By Lemma~\ref{lem:const-sig-layers}, for each query list, there are at most $\kappa^{2\kappa}$ values of $a$. In each, when we add new queries of length $r$, we first partition $r$ into the sum of $r_1\ldots r_i$ where the sum is $r$ and each $1\leq r_\iota < r$. Letting $N_{r'}$ denote the size of $\cF_{r'}$, this quantity is at most
    \[
        \prod_{\iota=1}^i |S_a|\cdot N_{r_\iota} \leq \left(\kappa^{4\kappa^2(r-1)^{r-1}}\right)^r\cdot \left(\kappa^{3\kappa^2}\right)^r \leq \kappa^{\kappa^2r^r}.
    \]

    The number of options for $r_1\ldots r_i$ and the number of significant layers $\kappa^{2\kappa}$ increases this by at most a factor of $\kappa^{\kappa^2}$.
\end{proof}

\begin{lemma} \label{lem:total-corr}
    The total maximum amount of corruption in $\del$, aside from layer $m$, is less than $2^{-\kappa}$.
\end{lemma}

\begin{proof}
    Between layers $0$ to $2\kappa^{\kappa}$, the amount of corruption in each layer is at most $\kappa^{-\kappa^2}$, so the total corruption is bounded by $2\kappa^\kappa \cdot \kappa^{-\kappa^2} < 2^{-\kappa-2}$.

    Afterwards, each query list in $\cF_r$ contributes at most $\kappa^{-\kappa^3 r^r}$ corruption for each of the $\kappa^{4\kappa^2 r^r}$ significant layers, and there are in total $\kappa^{\kappa^2r^r}$ queries. Multiplying these quantities, the total corruption is at most $\kappa^{-\kappa^3 r^r/2}$. Adding over $r=1$ to $k$, the total is less than $2^{-\kappa-1}$. The third type of deletion contributes at most $2^{-\kappa-2}$ corruption, and thus the total corruption is at most $2^{-\kappa}$.
\end{proof}

\subsection{Main Lemma} \label{sec:main-lemma}

Let $\cF$ be the family of representative queries and $\del$ the layered deletion pattern.

\begin{lemma} \label{lem:main}
    For any query difference list $\widehat{d_1}\ldots \widehat{d_{r-1}}$ for $r\leq k+1$, define $\ell:=\lfloor \log_\kappa (\widehat{d_1}+\ldots+ \widehat{d_{r-1}})\cdot \kappa^2 \rfloor$. If $\ell<m-1.5\kappa^\kappa$, there exists a member of $\cF_r$ such that $\widehat{\sd_1}\ldots \widehat{\sd_{r-1}}$ $2^{-\kappa/1000+r}$-approximates $\widehat{d_1}\ldots \widehat{d_{r-1}}$ on the deletion pattern $\del_\ell$ for at least $1-\kappa^{-100}$ fraction of $z\in \text{Im}(\C)$.
\end{lemma}

\begin{proof}

To show this claim, we'll induct on $r$.

For the base case of $r=1$, the only possible $0$ element list is the empty list, which is in $\cF_1$. Any single query $0$-approximates itself on all strings of length $M$.

We will now prove the statement for $r$ assuming it holds for all $r'<r$. Consider any query difference list $\widehat{d_1}\ldots \widehat{d_{r-1}}$. Throughout the proof, let's denote $c:= \lfloor \log_\kappa (\widehat{d_1}+\ldots+ \widehat{d_{r-1}}) \rfloor$. We separate the proof into three cases.

\paragraph{Case 1: $c<2\kappa^{\kappa^2}$.}

Consider the list of differences $\widehat{d_1}\ldots \widehat{d_{r-1}}$ each of which is less than $2^c$. In the first list of lists we added to $\cF$, we constructed a set $S$ of candidate differences. By the method by which the set was constructed, for every $\widehat{d_i} > \kappa^3$, there is $s\in S$ such that $\tau:=\widehat{d_i}-s < \kappa^{-2\kappa^2}$. Moreover, the layer $a:=\log(\widehat{d_i})-\kappa^{100}$ is larger than $0$, and has $\kappa^{-\kappa^2}>\kappa^{\kappa^3r^r}$ corruption, so we can apply Lemma~\ref{lem:corrupted-layers} to replace $\widehat{d_i}$ with $s$ in the query. Doing this for each $\widehat{d_i} > \kappa^3$, iteratively creates queries that are $2^{-\kappa/2}$ approximations. By Lemma~\ref{lem:transitivity}, the final query upon all the replacements is a $2^{-\kappa/2}k$ approximation, and since all the differences are now in $S$, the list is in $\cF$ by construction.

%\paragraph{Case 1: For all $c-\frac{\kappa^\kappa}{2} < a < c+\frac{\kappa^\kappa}{2}$ it holds that $f_\del(a) > \kappa^{-\kappa^3r^r}$.} 

\paragraph{Case 2: There exists $\widehat{\sd'_1}\ldots \widehat{\sd'_{r'}} \in \cF_{r'}$ for some $r'<r$ such that $a$ is a significant layer for the threshold $2^{-\kappa}$ for some $c-\frac{\kappa^\kappa}{10} < a < c+\frac{\kappa^\kappa}{10}$.} \mbox{}\\

The layer $a$ was corrupted with corruption at least $\kappa^{-\kappa^3 r^r}$. To show this, it suffices that $a$ was significant layer for $\widehat{\sd'_1}\ldots \widehat{\sd'_{r'}}$ when we reached layer $a$ in the definition process, not only for the final value of $\del$. $a$ is a significant layer for $\widehat{\sd'_1}\ldots \widehat{\sd'_{r'}}$ if $a$ is within $10\kappa^\kappa$ queries above $a':=\lfloor\sd'_1+\ldots+\sd'_{r'}\rfloor$. If not, then the first $a'$ layers have had their full corruptions finalized by the time layer $a$ had been reached, and whether something is a significant layer for queries $\widehat{\sd'_1}\ldots \widehat{\sd'_{r'}}$ depends only on the first $a'$ layers of corruption.

Now, we'll find the list $\widehat{\sd_1}\ldots \widehat{\sd_r} \in \cF_r$ that approximates $\widehat{d_1}\ldots \widehat{d_{r-1}}$.

Select a threshold layer $a-\frac{\kappa^\kappa}{10} < a' < a$ such that no value of $\log_\kappa{\widehat{d_i}}$ falls between $a'-\kappa^{\kappa/20}$ and $a'+\kappa^{\kappa/20}$. This exists because there are only $r\leq k$ differences and $\frac{\kappa^\kappa}{10}$ layers.

When describing the corruption process for $a$ as a significant layer, we constructed the set $S_a$ of differences. The set was constructed so that for every integer between $\kappa^{a-\kappa^\kappa}$ and $\kappa^{a+\kappa^\kappa}$, there is $s\in S_a$ that is at most a factor of $\kappa^{-2\kappa^2r^r}$ smaller. This requires that $\kappa^{a-\kappa^\kappa} > \kappa^{\kappa^2r^r}$ due to approximation issues. This holds as long as $a>\kappa^\kappa+{\kappa^2r^r}$, which is true since $c>2\kappa^{\kappa^2}$. 

Partition the list $\widehat{d_1}\ldots \widehat{d_{r-1}}$ into groups of consecutive elements as follows: starting from the first $\widehat{d_i}$, let a group be all the following consecutive elements satisfying $\widehat{d_i} < \kappa^{a'}$. If there are none, just let it be the next element that is larger. As such, groups of elements are either a single large element or a consecutive list of small elements. 

From left to right, if the next group of query differences is a single element larger than $\kappa^{a'}$, then by Lemma~\ref{lem:corrupted-layers}, for every $z \in \text{Im}(\C)$, there exists a way to replace it with an element of $S_a$ such that the new query approximates the old one on all strings of length $\kappa^m$.

If the next group of query differences (length $r'$) is a group of elements all smaller than $\kappa^{a'-\kappa^{\kappa/20}}$, then the conditions of Lemma~\ref{lem:replace-middle} are satisfied, so we can replace the group of query differences with one that approximates it in $\cF_{r'}$ by the inductive hypothesis. On all $z\in \text{Im}(\C)$ where the approximation held, which was $(1-(2k)^{r-1}2^{\kappa/2})$ fraction at least, the approximation holds for the new list of queries as well by Lemma~\ref{lem:replace-middle}.

In all, this new list of query differences has been constructed using elements of $\cF_{r'}$ and with differences in $S_a$. As such, this list of differences is in $\cF_{r}$ by construction.

Since each successive change to the query provided a $2^{-\kappa/100}k$-approximation of the previous step, this means that $\widehat{\sd_1}\ldots \widehat{\sd_{r-1}}$ $2^{-\kappa/100}k$-approximates the query $\widehat{d_1} \ldots \widehat{d_{r-1}}$ on $1-(2k)^r/2\cdot 2^{\kappa/2}$ of the values of $z\in \text{Im}(\C)$ by Lemma~\ref{lem:transitivity} and a union bound on the deletion pattern .

\paragraph{Case 3: For all $\widehat{\sq'_1}\ldots \widehat{\sq'_{r'}} \in \cF_{r'}$ for some $r'<r$, it holds that $a$ is not a significant layer for the threshold $2^{-\kappa}$ for any $c-\frac{\kappa^\kappa}{10} < a < c+\frac{\kappa^\kappa}{10}$.} \mbox{} \\

Define the layer $a$ to be the largest significant layer for some $\widehat{\sq'_1}\ldots \widehat{\sq'_{r'}} \in \cF_{r'}$ for some $r'<r$ that is smaller than $c$. 

Similar to the previous case, select a threshold layer $a+\frac{\kappa^\kappa}{50} < a' < a+\frac{\kappa^\kappa}{20}$ such that no value of $\log_\kappa{\widehat{d_i}}$ falls between $a'-\kappa^{\kappa/80}$ and $a'$. Then, for at least $2^{-\kappa/2}$ fraction of values of $z\in \text{Im}(C)$, the following holds by Markov's inequality and the definition of a significant layer.
    \[
        \mathop{\bbE}_{j \in \left[\frac{M}{\kappa^{a'-\kappa^2}}\right]} \left[ ~
        \left| 
        \begin{aligned}
            \freq{b}{C_{j',\kappa^{\widehat{d_1}+\ldots+\widehat{d_{r-1}}+\kappa^2}}(z)}{\widehat{d_1}\ldots \widehat{d_{r-1}}}{\del_{\tau | \widehat{d_1}+\ldots+\widehat{d_{r-1}}+\kappa^2}} \\
            - \freq{b}{C_{j,\kappa^{a'-\kappa^2}}(z)}{\widehat{d_1}\ldots \widehat{d_{r-1}}}{\del_{\tau | a'-\kappa^2}}
        \end{aligned}
        \right|
        ~ \right] < 2^{-\kappa/2}
    \]
    where $j' = \left\lfloor \frac{i}{\kappa^{\widehat{d_1}+\ldots+\widehat{d_{r-1}}+\kappa^2-a'-\kappa^2}} \right\rfloor$.

Now, we can apply Lemma~\ref{lem:no-sig-layers}, and we can replace all the $\widehat{d_i}$ that are large with precisely the value $\kappa^{a'}$. This new query $1-(2k)^r/2\cdot 2^{\kappa/2}$ approximates the old one on $1-(2k)^r/2\cdot2^{\kappa/2}$ values of $z\in \text{Im}(\C)$, and additionally now falls into Case 2. Therefore, there is a query that approximates this query, and ultimately $1-(2k)^r\cdot 2^{\kappa/2}$ approximates the original query on $1-(2k)^r\cdot2^{\kappa/2}$ values of $z\in \text{Im}(\C)$ on the deletion pattern $\del_a$.
\end{proof}

\begin{definition} [$\cG$ and $\cQ$]
    Let $\cG \subset [M]$ be the differences that appear in some list in $\cF$. This set $\cG$ is constant-sized. Moreover, add every multiple of $\kappa^{m-1.75\kappa^\kappa}$ that is at most $M$ to the set $\cG$.
    
    Let $\cQ$ be constructed as follows. For every tuple $g_1,\ldots g_k\in \cG$, let $\widehat{\sq_1}:=g_1, \widehat{\sq_2}:=g_1+g_2\ldots$, and add this list of queries to $\cQ$. This set $\cQ$ is constant-sized.
\end{definition}

The following lemma says that querying $\widehat{q_1}\ldots \widehat{q_k}$ will produce essentially the same output as some $\widehat{\sq_1}\ldots \widehat{\sq_k} \in \cQ$. 

\begin{lemma} \label{lem:main-cor}
    For any (ordered) query list $\widehat{q_1}\ldots \widehat{q_k} \in [M]^k$ and for at least $1-\kappa^{-4}$ fraction of $z\in \text{Im}(\C)$, there exists a list of queries $\widehat{\sq_1}\ldots \widehat{\sq_k} \in \cQ$ such that $\widehat{\sq_1} \ldots \widehat{\sq_k}$ approximates $\widehat{q_1}\ldots \widehat{q_k}$ for the deletion pattern $\del$, meaning the following. For any $b\in \Sigma^k$, 
    \[
        \Bigg| \bbP \Big[ z(\cD[\widehat{\sq_1}\ldots \widehat{\sq_k}]) = b \Big] - \bbP \Big[ z(\cD[\widehat{q_1}\ldots \widehat{q_k}]) = b \Big] \Bigg| <\kappa^{-4}.
    \]
\end{lemma}

\begin{proof}
    Given a list of queries $\widehat{q_1}\ldots \widehat{q_k} \in [M]^k$, separate the differences $\widehat{d_1}\ldots \widehat{d_{k-1}}$ as follows. Find a layer $a\in [m-1.7\kappa^\kappa,m-1.6\kappa^\kappa]$ such that no value of $\log_\kappa\widehat{d_i}$ falls between $a-\kappa^{\kappa/20}$ and $a+\kappa^{\kappa/20}$.

    Using this layer, we can designate the small and large differences. At a high level, we'll show that replacing sequences of small differences with ones that approximate them on some layer $c<a'$ does not change the output distribution by much, and adjusting large differences by $\kappa^{m-1.75\kappa^\kappa}$. The main difference between this situation and Case 2 of Lemma~\ref{lem:main} is that we won't desire to shift the entire query by uniformly random amount in $[\kappa^m]$; only by an amount in $[0,\kappa^{m-5}$, since that is the corruption of layer $m$.

    By Lemma~\ref{lem:main}, for at least $1-\kappa^{-100}$ fraction of the $z\in \text{Im}(\C)$, a given short sequence of differences can be $2^{-\kappa/2000}$-approximated by one in $\cG$. In particular, this means that for at least $1-\kappa^{-90}$ of the $z\in \text{Im}(\C)$, a $2^{\kappa/5000}$-approximation holds on all of the $\kappa^{m-10}$-sized chunks.

    The shift of the queries dictated by the deletion at layer $m$ is $\kappa^{m-5}$. Therefore, when a random shift (deletion at layer $m$ is chosen), except for probability of failure at most $\kappa^{-5}$, this deletion amount can be shifted by $\kappa^{m-10}$. In other words, except with probability at most $\kappa^{-5}$, the shift will be uniform within this chunk. 

    This allows us to apply Lemma~\ref{lem:replace-middle} on the layer $a$ to replace sequences of short queries with a sequence in $\cG$, and since the output given a uniform shift is a $2^{-\kappa/1000}$-approximation, it is in this situation as well, since except with probability $1-\kappa^{-5}$ the query could be simulated with a uniform shift. Next, we can apply Lemma~\ref{lem:corrupted-layers} to adjust the large queries slightly such that differences of successive queries are multiples of $\kappa^{m-1.75\kappa^\kappa}$. In total, the probability of this new query output not appropriately simulating the original $\widehat{q_1}\ldots \widehat{q_k}$ is at most $\kappa^{-4}$, proving the lemma.
\end{proof}

\mnote{Main notes to self about proof: $r^r$ can be removed, and in general, constants can likely be simplified. Proofs of 3 main lemmas and the last one need to be much more rigorous and should include factored out statements. Though i think the general structure is pretty competent.}

%We remark that in a deviation from the notation of $\cQ$ in the overview, $\cQ$ represents lists of differences of queries rather than singular individual queries.

% \begin{proof}
%     Extend the list of queries $\widehat{q_1}\ldots \widehat{q_k}$ to $1,\widehat{q_1}\ldots \widehat{q_k}$ and let $\widehat{d_1}:=\widehat{q_1}-1$, $\widehat{d_2}:=\widehat{q_2}-\widehat{q_1}$, and so on. Then, let $\widehat{\sd_1}\ldots \widehat{\sd_k} \in \cF_{k+1}$ be the of differences that $2^{-\kappa/101}$-approximates $\widehat{d_1}\ldots \widehat{d_k}$ on the layer $a:= \lfloor \log_\kappa (\widehat{d_1}+\ldots+ \widehat{d_{r-1}})\cdot \kappa^2 \rfloor$. for at least $0.99$ fraction of $z\in \text{Im}(\C)$. We can find this by Lemma~\ref{lem:main}. Then, a $2^{-\kappa/1000}$-approximation also holds on layer $m$ by Lemma~\ref{lem:approx-higher-layers}.

%     This implies that for all $b\in \Sigma^{k+1}$,
%     \begin{align*}
%         \begin{aligned}
%         \Big| \freq{b}{s}{\widehat{\sd_1}\ldots \widehat{\sd_{r-1}}}{\del} - \freq{b}{s}{\widehat{d_1}\ldots \widehat{d_{r-1}}}{\del} \Big|& \\
%         + \err(s;\widehat{\sd_1}\ldots \widehat{\sd_{r-1}};\del) + \err(s;\widehat{d_1}\ldots \widehat{d_{r-1}};\del)&
%         \end{aligned} ~< \delta,
%     \end{align*}
%     which in turn means that

% \end{proof}

\subsection{Concluding Theorem~\ref{thm:main} and Corollary~\ref{cor:lcc}.}

We begin with a lemma that shows that $z\in \text{Im}(\C)$ can be compressed down to a constant $\overline{K}(\eps, |\Sigma|, k)$ length string such that outputs of any list of $k$ queries can be predicted by this compressed string.

\begin{lemma} \label{lem:compression}
    There exists a constant $\overline{K}:=\overline{K}(\eps, |\Sigma|, k)$ and a compression function $f:\text{Im}(\C) \to \Sigma^{\overline{K}}$ such that the following holds. For any $k$ queries $\widehat{q_1}\ldots\widehat{q_k}$, there is a randomized function $\dec: \Sigma^{\overline{K}} \to \Sigma^k$ such that for at least $1-\kappa^{-1}$ fraction of $z\in \text{Im}(\C)$, the distribution of $\dec(f(z))$ has TV distance at most $1-\kappa^{-1}$ from the outputs of the queries $\widehat{q_1}\ldots\widehat{q_k}$ after applying the deletion pattern $\del$.
\end{lemma}

\begin{proof}
    This is essentially a restatement of Lemma~\ref{lem:main-cor}. In particular, define $f$ as follows. Apply the deletion pattern $\del$, and for every query $\widehat{\sq}$ in any list in $\cQ$, let the output of $\widehat{\sq}$ be included in the compression. Then, $\dec$ can be performed by taking the queries $\widehat{\sq_1}\ldots\widehat{\sq_k} \in \cQ$ that approximate $\widehat{q_1}\ldots\widehat{q_k}$, and looking at $f(z)$ to figure those out. The output of this has TV distance at most $1-\kappa^{-1}$ for at least $1-\kappa^{-1}$ fraction of values of $z$.
\end{proof}

\begin{lemma} \label{lem:index-from-Fk}
    If the output of all lists of queries can be deduced up to TV distance $1-\kappa^{-1}$ on $\kappa^{-1}$-fraction of strings from some randomized compression function $f:\text{Im}(\C) \to \Sigma^{\overline{K}}$, there exists a constant $K$ so that for $n>K$, there exists $i\in [n]$ such that for uniformly random $x$, the index $x[i]$ cannot be guessed correctly from any query list with probability more than $\frac12+\eps^2$.
\end{lemma}

\begin{proof}
    Assume for the sake of contradiction that the queries are capable of recovering every index $i$ with probability $\frac12+\eps^2$. Then, on $1-\kappa^{-1}$-fraction of inputs, the compressed version of it $f(\C(x))$ is capable of recovering $i$ with probability $\frac12+\eps+\kappa^{-1}$. 
    
    Perform the following process, which we'll call $Y$: take a random string $x\in \Sigma^n$, take the compressed string $f(\C(x))$ and use it to recover each of $x[1]\ldots x[n]$. We know that $H(Y)\leq \overline{K}$. Each index $i$ has a $\frac12+\eps^2-2\kappa^{-1}>\frac12+\eps^3$ chance (but not independent) of yielding the correct value of $x[i]$. In total, this means the entropy of each bit of the string $x$ conditional on the guesses provided by the string is at most $1-\eps^6$. By the subadditivity of entropy, the entropy of this is at most $(1-\eps^6)n$, so $H(Y)\leq (1-\eps^6)n$. Also, the entropy of the string $x$, which we'll call $H(X)$, is $n$.

    We know that 
    \begin{align*}
        H(X)&\leq H(X|Y)+H(Y) \\
        \implies n &\leq \overline{K} + (1-\eps^6)n
    \end{align*}
    which is a contradiction for sufficiently large $n$.
\end{proof}

Combining Lemma~\ref{lem:compression} and Lemma~\ref{lem:index-from-Fk} concludes the proof of Theorem~\ref{thm:main}.

Theorem~\ref{cor:lcc} follows by Theorem 4 in \cite{BlockiCGLZZ21}, which states that the existence of LCC's in the adversarial insertion/deletion setting would imply LDC's in the adversarial insertion/deletion setting, which don't exist by Theorem~\ref{thm:main}. 

Our result requires that deletion resilient LCC's would imply deletion resilient LDC's. Their proof carries over identically to this setting, and we briefly sketch why.

\begin{proof}[Proof Sketch] The main property of the corruption pattern they used was that no two codewords $z,z'$ in the codespace of an LCC could have Hamming distance less than $\eps' n$, where $\eps'$ is the allotted corruption for the adversary. This also holds in the deletion-only setting, because if two codewords $z$ and $z'$ were within Hamming distance $\eps' n$, the adversary could delete all the indices on which they differ. The resulting messages would be the same, and any bit that differed between the two encodings cannot be recovered.

Then, the codespace has VC dimension $d=O_{\eps,|\Sigma|}(n)$, by Lemma 6 in \cite{BlockiCGLZZ21}. Restrict the code to a subset $C'$ of size $\Sigma^d$ with the same shattering set, and define the corresponding code $\Sigma^d \to C'$, as the extension of the values on the shattering set into $C'$. This new code is an LCC that follows a systematic encoding, which would also make it an LDC on messages of length $d=O_{\eps,|\Sigma|}(n)$. This cannot exist by Theorem~\ref{thm:main}.
\end{proof}
\section{Acknowledgements}

Meghal Gupta is supported by an NSF Graduate Research Fellowship. The author would like to thank her advisor Dr. Venkatesan Guruswami for invaluable guidance, helpful discussions, and feedback on the manuscript. The author would also like to thank Omar Alrabiah and Yang Liu for helpful discussions and comments on the paper.

\bibliographystyle{alpha}
\bibliography{refs}

\end{document}